\documentclass[envcountsame]{llncs}

\usepackage{amssymb}
\usepackage{amsmath}
\usepackage{mathrsfs}
\usepackage{txfonts}

\newcommand{\calF}{\mathcal{F}}

\newcommand{\calP}{\mathcal{P}}

\newcommand{\A}{\mathscr{A}}
\newcommand{\B}{\mathscr{B}}

\newcommand{\E}{\mathbb{E}}
\newcommand{\calE}{\mathcal{E}}
\newcommand{\F}{\mathcal{F}}
\newcommand{\G}{\mathcal{G}}
\newcommand{\R}{\mathcal{R}}

\newcommand{\calL}{\mathcal{L}}

\newcommand{\V}{\mathcal{V}}

\newcommand{\M}{\mathcal{M}}

\newcommand{\X}{\mathcal{X}}
\newcommand{\Y}{\mathcal{Y}}
\newcommand{\Prob}{\mathit{Prob}}
\newcommand{\Nset}{\mathbb{N}}
\newcommand{\Nseto}{\Nset_0}

\newcommand{\Zset}{\mathbb{Z}}
\newcommand{\Qset}{\mathbb{Q}}
\newcommand{\Rset}{\mathbb{R}}
\newcommand{\fpath}{\mathit{FPath}}
\newcommand{\run}{\mathit{Run}}

\newcommand{\len}{\mathit{length}}
\newcommand{\pre}{\mathit{Pre}}
\newcommand{\post}{\mathit{Post}}

\newcommand{\norm}[1]{\|#1\|}

\newcommand{\pfterm}{T^{>0}_{<\infty}}
\newcommand{\pterm}{T^{>0}}

\newcommand{\tran}[1]{{}\mathchoice%
    {\stackrel{#1}{\rightarrow}}
    {\mathop {\smash\rightarrow}\limits^{\vrule width 0pt height 0pt
                                                depth 4pt\smash{#1}}}
    {\stackrel{#1}{\rightarrow}}
    {\stackrel{#1}{\rightarrow}}
{}}

\newcommand{\prule}[1]{{}\mathchoice%
    {\stackrel{#1}{\longrightarrow}}
    {\mathop {\smash\longrightarrow_{>0}}\limits^{\vrule width 0pt height 0pt
          depth 2.5pt\smash{#1\hspace*{5pt}}}}
    {\stackrel{#1}{\longrightarrow}_{>0}}
    {\stackrel{#1}{\longrightarrow}_{>0}}
{}}
\newcommand{\zrule}[1]{{}\mathchoice%
    {\stackrel{#1}{\longrightarrow_z}}
    {\mathop {\smash\longrightarrow_{=0}}\limits^{\vrule width 0pt height 0pt
          depth 2.5pt\smash{#1\hspace*{5pt}}}}
    {\stackrel{#1}{\longrightarrow}_{=0}}
    {\stackrel{#1}{\longrightarrow}_{=0}}
{}}

\usepackage{color}
\definecolor{orange3}{rgb}{1.0,0.2538,0.1681}
\definecolor{blau}{rgb}{0,.39608,0.74118}
\definecolor{rot}{rgb}{0.79216,.12941,0.24706}
\definecolor{ml}{rgb}{0.2,0.7,0.7}

\newcommand\Tony[1]{}
\newcommand\Stefan[1]{}
\newcommand\Tomas[1]{}
\newcommand{\tony}[1]{}
\newcommand\stefan[1]{}
\newcommand\tomas[1]{}

\newcommand{\vzero}{\boldsymbol{0}}
\newcommand{\vone}{\boldsymbol{1}}
\newcommand{\vv}{\boldsymbol{v}}
\newcommand{\vd}{\boldsymbol{d}}%
\newcommand{\vu}{\boldsymbol{u}}%
\newcommand{\vmax}{\vv_{\max}}
\newcommand{\vmin}{\vv_{\min}}
\newcommand{\va}{|\vv|}
\newcommand{\vs}{\boldsymbol{s}}
\newcommand{\valpha}{\boldsymbol{\alpha}}

\newcommand{\xmin}{x_{\min}}

\newcommand{\cs}[1]{c^{(#1)}}
\newcommand{\ps}[1]{p^{(#1)}}
\newcommand{\ms}[1]{m^{(#1)}}
\newcommand{\acc}{\mathit{acc}}
\newcommand{\rej}{\mathit{rej}}
\newcommand{\cons}{\mathit{cons}}
\newcommand{\inco}{\mathit{inco}}

\newcommand{\qedv}[1]{\ifmmode\squareforqed\else{\unskip\nobreak\hfil
\penalty50\hskip1em\null\nobreak\hfil(#1)\squareforqed
\parfillskip=0pt\finalhyphendemerits=0\endgraf}\fi}

\newcommand{\kw}[1]{\textbf{\rmfamily #1}}
\newcommand{\msp}{\hspace*{1.5em}}

\newcommand{\theoremlike}[2]{\par\medskip\penalty-250%
{{\bfseries\noindent
#2 \ref{#1}.}}\it}

\newcommand{\thmhelperpre}[2]{\theoremlike{#1}{#2}}
\newcommand{\thmhelperpost}{\par\medskip}

\newenvironment{reftheorem}[1]{\thmhelperpre{#1}{Theorem}}{\thmhelperpost}

\newenvironment{refproposition}[1]{\thmhelperpre{#1}{Proposition}}{\thmhelperpost }

\begin{document}

\title{Efficient Analysis of Probabilistic Programs with
an Unbounded Counter%\thanks {Supported by the
%    research center Institute for Theoretical Computer Science (ITI,
%    project No.~1M0545) and by Czech Science Foundation, grant
%    No.~P202/10/1469.}
}

\author{Tom\'{a}\v{s} Br\'{a}zdil\inst{1} \and
        Stefan Kiefer\inst{2} \and
        Anton\'{\i}n Ku\v{c}era\inst{1}}

\institute{Faculty of Informatics, Masaryk University, Czech Republic.\\
    \texttt{\{brazdil,kucera\}@fi.muni.cz} \and
     Oxford University Computing Laboratory, United Kingdom.\\
  \texttt{stefan.kiefer@comlab.ox.ac.uk}}

\maketitle

\begin{abstract}
\noindent
We show that a subclass of infinite-state probabilistic programs
that can be modeled by probabilistic one-counter automata (pOC)
admits an efficient quantitative analysis. In particular, we show
that the expected termination time can be approximated up to an arbitrarily
small relative error with polynomially many arithmetic operations, and the same holds
for the probability of all runs that satisfy a given $\omega$-regular
property.
Further, our results establish a powerful link between pOC
and martingale theory, which leads to fundamental
observations about quantitative properties of runs in pOC.
In particular, we provide a ``divergence gap theorem'', which bounds
a positive non-termination probability in pOC away from zero.
\end{abstract}

\section{Introduction}
\label{sec-intro}

In this paper we aim at designing \emph{efficient} algorithms for
analyzing basic properties of probabilistic programs operating
on unbounded data domains that can be abstracted into
a non-negative integer counter. Consider, e.g., the recursive
program of Fig.~\ref{fig-and-or} which evaluates a given AND-OR
tree, i.e., a tree whose root is an AND node, all descendants
of AND nodes are either leaves or OR nodes, and all descendants
of OR nodes are either leaves or AND nodes. Note that the program
evaluates a subtree only when necessary. In general, the program
may not terminate and we cannot say anything about its expected
termination time. Now assume that we \emph{do} have some knowledge about
the actual input domain of the program, which might have been gathered
empirically:
\begin{itemize}
\item an AND node has about $a$ descendants on average;
\item an OR node has about $o$ descendants on average;
\item the length of a branch is $b$ on average;
\item the probability that a leaf evaluates to $1$ is $z$.
\end{itemize}
Further, let us assume that the actual number of descendants and the
actual length of a branch are \emph{geometrically} distributed (which
is a reasonably good approximation in many cases). Hence, the probability
that an AND node has \emph{exactly} $n$ descendants is
$(1-x_a)^{n-1} x_a$ with $x_a = \frac{1}{a}$. Under these assumption, the
behaviour of the program is well-defined in the probabilistic sense,
and we may ask the following questions:
\begin{itemize}
\item[1)] Does the program terminate with probability one? If not,
  what is the termination probability?
\item[2)] If we restrict ourselves to terminating runs, what is the
  expected termination time? (Note that this conditional expected
  value is defined even if our program does not terminate with probability
  one.)
% \item If the program does not terminate with a positive probability,
%   how long should we wait before killing the program so that the
%   chance of killing a terminating run is less that $1\%$?
\end{itemize}
These questions are not trivial, and at first glance it is
not clear how to approach them.
Apart of the expected termination time,
which is a fundamental characteristic of terminating runs, we are
also interested in the properties on \emph{non-terminating} runs,
specified by linear-time logics or automata on infinite words. Here,
we ask for the probability of all runs satisfying a given
linear-time property. Using the results of this paper, answers to such
questions can be computed \emph{efficiently} for a large class of
programs, including the one of Fig.~\ref{fig-and-or}. More precisely,
the first question about the probability of termination can be answered
using the existing results \cite{EWY:one-counter}; the original
contributions of this paper are efficient algorithms for computing
answers to the remaining questions.

\begin{figure}[t]
\noindent\hspace{\fill}
\parbox{.45\textwidth}{\ttfamily\footnotesize
\kw{procedure} AND(node)\\[.5ex]
\kw{if} node is a leaf\\
\msp   \kw{then} \kw{return} node.value\\
\kw{else}\\
\msp   \kw{for each} successor s of node \kw{do}\\
\msp\msp      \kw{if} OR(s) = 0 \kw{then} \kw{return} 0\\
\msp   \kw{end for}\\
\msp   \kw{return} 1\\
\kw{end if}}\hspace{\fill}
\parbox{.45\textwidth}{\ttfamily\footnotesize
\kw{procedure} OR(node)\\[.5ex]
\kw{if} node is a leaf\\
\msp   \kw{then} \kw{return} node.value\\
\kw{else}\\
\msp   \kw{for each} successor s of node \kw{do}\\
\msp\msp      \kw{if} AND(s) = 1 \kw{then} \kw{return} 1\\
\msp   \kw{end for}\\
\msp   \kw{return} 0\\
\kw{end if}}\hspace{\fill}
\caption{A recursive program for evaluating AND-OR trees.}
\label{fig-and-or}
\end{figure}

The abstract class of probabilistic programs considered in this paper
corresponds to \emph{probabilistic one-counter automata (pOC)}.
Informally, a pOC has finitely many control states $p,q,\ldots$ that
can store global data, and a single non-negative counter that can be
incremented, decremented, and tested for zero.  The dynamics of a
given pOC is described by finite sets of \emph{positive} and
\emph{zero} rules of the form $p \prule{x,c} q$ and $p \zrule{x,c} q$,
respectively, where $p,q$ are control states, $x$ is the
\emph{probability} of the rule, and $c \in \{-1,0,1\}$ is the
\emph{counter change} which must be non-negative in zero rules. A
\emph{configuration} $p(i)$ is given by the current control state $p$
and the current counter value~$i$.  If $i$ is positive/zero, then
positive/zero rules can be applied to $p(i)$ in the natural way. Thus,
every pOC determines an infinite-state Markov chain where states are
the configurations and transitions are determined by the rules. As an
example, consider a pOC model of the program of
Fig.~\ref{fig-and-or}. We use the counter to abstract the stack of
activation records. Since the procedures AND and OR alternate
regularly in the stack, we keep just the current stack height in the
counter, and maintain the ``type'' of the current procedure
in the finite control
(when we increase or decrease the counter, the ``type'' is
swapped). The return values of the two procedures are also stored in
the finite control. Thus, we obtain the pOC model of
Fig.~\ref{fig-and-or-model} with $6$ control states and $12$ positive
rules (zero rules are irrelevant and hence not shown in
Fig.~\ref{fig-and-or-model}). The initial configuration is
$(\textit{and,init})(1)$, and the pOC terminates either in
$(\textit{or,return,0})(0)$ or $(\textit{or,return,1})(0)$, which
corresponds to evaluating the input tree to $0$ and $1$, respectively.
We set $x_a := 1/a$, $x_o := 1/o$ and $y := 1/b$
 in order to obtain the average numbers $a, o, b$ from the beginning.

\begin{figure}[t]
\noindent\hspace{\fill}
\parbox{.46\textwidth}{\footnotesize
\textrm{/* if we have a leaf, return 0 or 1 */}\\
\msp$(\textit{and,init}) \xrightarrow{y\, z,\,-1} (\textit{or,return,1})$,\\
\msp$(\textit{and,init}) \xrightarrow{y (1-z),\,-1}
    (\textit{or,return,0})$\\[.5ex]
\textrm{/* otherwise, call OR */}\\
\msp$(\textit{and,init}) \xrightarrow{(1-y),1}   (\textit{or, init})$\\[.5ex]
\textrm{/* if OR returns 1, call another OR? */}\\
\msp$(\textit{and,return,1}) \xrightarrow{(1-x_a),\,1}  (\textit{or,init})$\\
\msp$(\textit{and,return,1}) \xrightarrow{x_a,\,-1}
   (\textit{or,return,1})$\\[.5ex]
\textrm{/* if OR returns 0, return 0 immediately */}\\
\msp$(\textit{and,return,0}) \xrightarrow{1,\,-1} (\textit{or,return,0})$}
\hspace{\fill}
\parbox{.46\textwidth}{\footnotesize
\textrm{/* if we have a leaf, return 0 or 1 */}\\
\msp$(\textit{or,init}) \xrightarrow{y\,z,\,-1} (\textit{and,return,1})$,\\
\msp$(\textit{or,init}) \xrightarrow{y(1-z),\,-1}
    (\textit{and,return,0})$\\[.5ex]
\textrm{/* otherwise, call AND */}\\
\msp$(\textit{or,init}) \xrightarrow{(1-y),1}   (\textit{and,init})$\\[.5ex]
\textrm{/* if AND returns 0, call another AND? */}\\
\msp$(\textit{or,return,0}) \xrightarrow{(1-x_o),\,1}  (\textit{and, init})$\\
\msp$(\textit{or,return,0}) \xrightarrow{x_o,\,-1}
   (\textit{and,return,0})$\\[.5ex]
\textrm{/* if AND returns 1, return 1 immediately */}\\
\msp$(\textit{or,return,1}) \xrightarrow{1,\,-1} (\textit{and,return,1})$}
\hspace{\fill}
\caption{A pOC model for the program of Fig.~\ref{fig-and-or}.}
\label{fig-and-or-model}
\end{figure}

As we already indicated, pOC can model recursive
programs operating on unbounded data structures such as trees, queues,
or lists, assuming that the structure can be faithfully abstracted
into a counter.
%  In particular, this applies in situations when there
% is only one procedure which calls itself recursively, and
% the body of the procedure is a case-list of the form
% \[
%   \textit{guard}_1: \textit{command}_1, \cdots,
%   \textit{guard}_n: \textit{command}_n
% \]
% where the probability of satisfaction of the individual guards is
% known (or can be reasonably estimated).
% \tomas{I think that we are not able to deal with {\bf recursive} programs with unbounded data, right? Either data or recursion but not both :-)}
Let us note that modeling general recursive
programs requires more powerful formalisms such as
\emph{probabilistic pushdown automata (pPDA)} or
\emph{recursive Markov chains (RMC)}. However, as it is
mentioned below,
pPDA and RMC do not admit \emph{efficient} quantitative analysis
for fundamental reasons. Hence, we must inevitably sacrifice
a part of pPDA modeling power to gain efficiency in algorithmic
analysis, and pOC seem to be a convenient compromise for achieving this
goal.

The relevance of pOC is not limited just
to recursive programs. As observed in \cite{EWY:one-counter},
pOC are equivalent,
in a well-defined sense, to discrete-time \emph{Quasi-Birth-Death
processes (QBDs)},
a well-established stochastic model that has been deeply studied since
late~60s.
Thus, the applicability of pOC extends to queuing theory, performance
evaluation, etc., where  QBDs are considered as a fundamental
formalism. Very recently, games over (probabilistic) one-counter
automata, also called ``energy games'', were considered in several
independent works \cite{CHD:energy-games,CHDHR:energy-mean-payoff,%
BBEKW:OC-MDP,BBE:OC-games}. The study is motivated by optimizing the use
of resources (such as energy) in modern computational devices.

\textbf{Previous work.} In \cite{EKM:prob-PDA-PCTL,EY:RMC-SG-equations},
it has been shown that the vector of termination probabilities in
pPDA and RMC is the least solution of an effectively constructible
system of quadratic equations. The termination probabilities may take
irrational values, but can be effectively approximated up to an
arbitrarily small absolute error $\varepsilon >0$
in polynomial space by employing the decision procedure for the
existential fragment of Tarski algebra (i.e., first order theory of
the reals) \cite{Canny:Tarski-exist-PSPACE}. Due to the results of
\cite{EY:RMC-SG-equations}, it is possible to approximate termination probabilities
in pPDA and RMC ``iteratively'' by using the decomposed Newton's method.
However, this approach may need exponentially many iterations of the
method before it starts to produce one bit of precision per iteration
\cite{KLE07:stoc}.
Further, any non-trivial approximation of the non-termination probabilities
is at least as hard as the \textsc{SquareRootSum}
problem~\cite{EY:RMC-SG-equations}, whose exact
complexity is a long-standing open question in exact numerical computations
(the best known upper bound for \textsc{SquareRootSum} is PSPACE).
Computing termination probabilities in pPDA and RMC up to a given
\emph{relative} error $\varepsilon > 0$, which is more relevant from
the point of view of this paper, is \emph{provably} infeasible because
the termination probabilities can be doubly-exponentially small
in the size of a given pPDA or RMC~\cite{EY:RMC-SG-equations}.

The expected termination time and the expected reward per transition
in pPDA and RMC has been studied in \cite{EKM:prob-PDA-expectations}.
In particular, it has been
shown that the tuple of expected termination times is the least solution
of an effectively constructible system of linear equations, where the
(products of) termination probabilities are used as coefficients. Hence,
the equational system can be represented only symbolically, and the
corresponding approximation algorithm again employs
the decision procedure for Tarski algebra. There also other results
for pPDA and RMC, which concern model-checking problems for
linear-time \cite{EY:RMC-LTL-complexity,EY:RMC-LTL-QUEST} and
branching-time \cite{BKS:pPDA-temporal} logics, long-run average
properties \cite{BEK:prob-PDA-predictions}, discounted
properties of runs \cite{BBHK:pPDA-discounted}, etc.

\textbf{Our contribution.}
In this paper, we build on the previously established results for pPDA
and RMC, and on the recent results of \cite{EWY:one-counter}
where is shown that the  decomposed Newton method
of \cite{KLE:Newton-STOC} can be used to compute
termination probabilities in pOC up to a given \emph{relative} error
$\varepsilon>0$ in time which is \emph{polynomial} in the size of pOC and
$\log(1/\varepsilon)$, assuming the unit-cost rational arithmetic
RAM (i.e., Blum-Shub-Smale) model of computation. Adopting
the same model, we show the following:
\begin{itemize}
\item[1.] The expected termination time in a pOC $\A$ is computable up to
  an arbitrarily small relative error $\varepsilon > 0$ in time
  polynomial in $|\A|$ and $\log(1/\varepsilon)$. Actually, we can even
  compute the expected termination time up to an arbitrarily small
  \emph{absolute} error, which is a better estimate because
  the expected termination time is always at least~$1$.
\item[2.] The probability of all runs in a pOC $\A$ satisfying
  an $\omega$-regular property encoded by a deterministic Rabin
  automaton $\R$  is computable up to
  an arbitrarily small relative error $\varepsilon > 0$ in time
  polynomial in $|\A|$, $|\R|$, and $\log(1/\varepsilon)$.
\end{itemize}
The crucial step towards obtaining these results is the construction
of a suitable \emph{martingale}  for a given
pOC, which allows to apply powerful results of martingale theory
(such as the optional stopping theorem or Azuma's inequality, see,
e.g., \cite{Rosenthal:book,Williams:book}) to
the quantitative analysis of pOC. In particular, we use this martingale
to establish the crucial \emph{divergence gap theorem} in
Section~\ref{sec-LTL}, which bounds a positive divergence
probability in pOC away from~$0$. The divergence gap theorem is
indispensable in analysing properties of non-terminating runs,
and together with the constructed martingale provide generic tools
for designing efficient approximation
algorithms for other interesting quantitative properties of~pOC.

Although our algorithms have polynomial worst-case complexity, the
obtained bounds look complicated and it is not immediately clear
whether the algorithms are practically usable. Therefore,
we created a simple experimental implementation which computes
the expected termination time for pOC, and used this tool to analyse the
pOC model of Fig.~\ref{fig-and-or-model}. The details are given
in Section~\ref{sec-concl}.

% Due to space limits, we could not include most of the proofs into the
% main body of the paper. These can be found in a full version of this
% paper \cite{BKK:pOC-termination-LTL-full}.

%%% Local Variables:
%%% mode: latex
%%% TeX-master: "main"
%%% End:

\section{Definitions}
\label{sec-defs}

\noindent
We use $\Zset$, $\Nset$, $\Nset_0$, $\Qset$, and $\Rset$
to denote the set of all integers, positive integers, non-negative
integers, rational numbers, and real numbers, respectively.
Let $\delta > 0$, $x \in \Qset$, and $y \in \Rset$.
We say that $x$ approximates $y$ up to
a relative error $\delta$, if either \mbox{$y \neq 0$} and
\mbox{$|x-y|/|y| \leq \delta$}, or $x = y = 0$. Further, we
say that $x$ approximates $y$ up to
an absolute error $\delta$ if $|x-y| \leq \delta$.
We use standard notation for intervals, e.g.,
$(0,1]$ denotes \mbox{$\{x \in \Rset \mid 0 < x \leq 1 \}$}.
% The set of finite words over an alphabet $\Sigma$ is denoted by
% $\Sigma^*$, and the set of infinite words over $\Sigma$ is denoted by
% $\Sigma^\omega$.  We also use $\Sigma^+$ to
% denote the set $\Sigma^* \smallsetminus \{\varepsilon\}$, where
% $\varepsilon$ is the empty word. The length of a given $w \in \Sigma^* \cup
% \Sigma^{\omega}$ is denoted by $\len(w)$, where the length of an infinite
% word is $\infty$.
% Given a word (finite or infinite) over $\Sigma$, the individual
% letters of $w$ are denoted by $w(0),w(1),\ldots$

Given a finite set~$Q$, we regard elements of~$\Rset^Q$ as vectors over~$Q$.
We use boldface symbols like $\vu, \vv$ for vectors.
In particular we write~$\vone$ for the vector whose entries are all~$1$.
Similarly, matrices are elements of~$\Rset^{Q \times Q}$.

Let $\V = (V,\tran{})$, where $V$ is a non-empty set of vertices and
${\tran{}} \subseteq V \times V$ a \emph{total} relation
(i.e., for every $v \in V$ there is some $u \in
V$ such that $v \tran{} u$). The reflexive and transitive closure of
$\tran{}$ is denoted by $\tran{}^*$. A \emph{finite path} in $\V$ of
\emph{length} $k \geq 0$ is a finite sequence
of vertices $v_0,\ldots,v_k$, where
$v_i \tran{} v_{i+1}$ for all $0 \leq i <k$. The length of a finite
path $w$ is denoted by $\len(w)$. A \emph{run} in $\V$ is an infinite
sequence $w$ of vertices such that every finite prefix of $w$ is
a finite path in $\V$. The individual vertices of $w$ are denoted by
$w(0),w(1),\ldots$ The sets of all finite paths and all runs
in $\V$ are denoted by $\fpath_{\V}$ and $\run_{\V}$, respectively.
The sets of all finite paths and all runs in $\V$
that start with a given finite path $w$ are denoted by
$\fpath_{\V}(w)$ and $\run_{\V}(w)$, respectively.
A \emph{bottom strongly connected component (BSCC)} of $\V$ is a
subset $B \subseteq V$ such that for all $v,u \in B$
we have that $v \tran{}^* u$, and whenever $v \tran{} u'$ for some $u' \in V$,
then $u' \in B$.
% Let $U \subseteq V$.
% %We say that $U$ is \emph{strongly connected}
% %if either $U = \{v\}$ where $v \tran{} v$, or for all $u,v \in V$
% %where $u \neq v$ we have that $u \tran{}^* v$.
% We say that $U$ is \emph{strongly connected}
% if $v\tran{}^+ u$ for all $v,u\in U$ (here $v\tran{}^+ u$ if there is a path
% of length greater than $1$ from $v$ to $u$).
% %either $U = \{v\}$ where $v \tran{} v$, or for all $u,v \in V$
% %where $u \neq v$ we have that $u \tran{}^* v$.
% %
% Further, we say that $U$
% is a \emph{strongly connected component (SCC)}  if $U \neq \emptyset$
% is a maximal strongly connected subset of $V$, and $U$ is
% a \emph{bottom SCC (BSCC)} if
% for every $u \in U$ and every $u \tran{} v$ we have that $v \in U$.

We assume familiarity with basic notions of probability theory, e.g.,
\emph{probability space}, \emph{random variable}, or the \emph{expected
value}. As usual, a \emph{probability distribution} over a finite or
countably infinite
set $X$ is a function
$f : X \rightarrow [0,1]$ such that \mbox{$\sum_{x \in X} f(x) = 1$}.
We call $f$ \emph{positive} if
$f(x) > 0$ for every $x \in X$, and \emph{rational} if $f(x) \in
\Qset$ for every $x \in X$.
%A \emph{$\sigma$-field} over a set $\Omega$ is a set
%$\calF \subseteq 2^{\Omega}$ that includes $\Omega$ and is closed under
%complement and countable union. A
%\emph{probability space} is a triple $(\Omega,\calF,\calP)$ where
%$\Omega$ is a set called \emph{sample space},
%$\calF$ is a $\sigma$-field over $\Omega$ whose elements are called
%\emph{events}, and $\calP : \calF \rightarrow [0,1]$ is a
%\emph{probability measure} such that, for each countable collection
%$\{X_i\}_{i\in I}$ of pairwise disjoint elements of $\calF$ we have that
%$\calP(\bigcup_{i\in I} X_i) = \sum_{i\in I} \calP(X_i)$, and moreover
%$\calP(\Omega)\mathord{=}1$.
% A \emph{random variable} over
% a probability space $(\Omega,\calF,\calP)$
% is a function  $X : \Omega \rightarrow \Rset$ such that
% \mbox{$\{\omega \in \Omega \mid X(\omega) \leq c\} \in \F$}
% for every $c \in \Rset$. The \emph{expected value}
% of a random variable $X$ is defined by
% $\expect{X} = \int_{\omega \in \Omega} X(\omega)\, \mathrm{d} \calP$.

\begin{definition}
  A \emph{Markov chain} is  a
  triple \mbox{$\M = (S,\tran{},\Prob)$}
  where $S$ is a finite or countably infinite set of \emph{states},
  ${\tran{}} \subseteq S \times S$ is a total \emph{transition relation},
  and $\Prob$ is a function that assigns to each state $s \in S$
  a positive probability distribution over the outgoing transitions
  of~$s$. As usual, we write $s \tran{x} t$ when $s \tran{} t$
  and $x$ is the probability of $s \tran{} t$.
\end{definition}
A Markov chain $\M$ can be also represented by its
\emph{transition matrix} $M \in [0,1]^{S{\times}S}$, where
$M_{s,t} = 0$ if $s \not\rightarrow t$, and $M_{s,t} = x$ if $s \tran{x} t$.

To every $s \in S$ we associate the probability
space $(\run_{\M}(s),\calF,\calP)$ of runs starting at $s$,
where
$\calF$ is the \mbox{$\sigma$-field} generated by all \emph{basic cylinders},
$\run_{\M}(w)$, where $w$ is a finite path starting at~$s$, and
$\calP: \calF \rightarrow [0,1]$ is the unique probability measure such that
$\calP(\run_{\M}(w)) = \prod_{i{=}1}^{\len(w)} x_i$ where
$w(i{-}1) \tran{x_i} w(i)$ for every $1 \leq i \leq \len(w)$.
If $\len(w) = 0$, we put $\calP(\run_{\M}(w)) = 1$.

\begin{definition}
\label{def-pOC}
A \emph{probabilistic one-counter automaton (pOC)} is a tuple,
$\A = (Q,\delta^{=0},\delta^{>0},P^{=0},P^{>0})$, where
\begin{itemize}
  \item $Q$ is a finite set of \emph{states},
  \item
    $\delta^{>0} \subseteq Q \times \{-1,0,1\} \times Q$
%    $\delta^{>0} \subseteq Q \times \{-1,0,1,\reset\} \times Q$
    and $\delta^{=0} \subseteq Q \times \{0,1\} \times Q$ are the sets
    of \emph{positive} and \emph{zero rules} such that each
    $p \in Q$ has an outgoing positive rule and an outgoing zero rule;
  \item $P^{>0}$ and $P^{=0}$ are \emph{probability assignments}:
    both assign to each $p \in Q$, a positive rational probability
    distribution over the outgoing
    rules in $\delta^{>0}$ and $\delta^{=0}$,
    respectively, of~$p$.
\end{itemize}
  % A \emph{configuration} of $\Delta$ is an element of $Q \times \Zset$,
% and the set of all configurations of $\Delta$ is denoted $\configs{\Delta}$.
\end{definition}

\noindent
In the following, we often write $p \zrule{x,c} q$ to denote that
$(p,c,q) \in \delta^{=0}$ and $P^{=0}(p,c,q) = x$, and similarly
$p \prule{x,c} q$ to denote that
$(p,c,q) \in \delta^{>0}$ and $P^{>0}(p,c,q) = x$. The size of $\A$,
denoted by $|\A|$, is the length of the string which represents~$\A$,
where the probabilities of rules are written in binary.
A \emph{configuration} of $\A$ is an element of $Q \times \Nseto$,
written as $p(i)$. To $\A$ we associate an infinite-state Markov
chain $\M_\A$ whose states are the configurations of $\A$, and for all
$p,q \in Q$, $i\in \Nset$, and $c \in \Nseto$ we have that
$p(0) \tran{x} q(c)$ iff $p \zrule{x,c} q$, and
$p(i) \tran{x} q(c)$ iff $p \prule{x,c{-}i} q$.
For all $p,q \in Q$, let
\begin{itemize}
\item $\run_\A(p{\downarrow}q)$ be the set of all
  runs in $\M_\A$ initiated in $p(1)$ that visit $q(0)$ and the counter
  stays positive in all configurations preceding this visit;
%\item $\run(p,q)$ be the set of all
%  runs initiated in $p(1)$ that visit a configuration
%  $q(i)$, where $i \in \Nseto$;
  %, and the counter
  %stays positive in all configurations preceding this visit;
\item $\run_\A(p{\uparrow})$ be the set of all runs in $\M_\A$
  initiated in $p(1)$ where the counter never reaches zero.
\end{itemize}
We omit the ``$\A$'' in $\run_\A(p{\downarrow}q)$ and
$\run_\A(p{\uparrow})$ when it is clear from the context, and
we use $[p{\downarrow}q]$ and $[p{\uparrow}]$ to denote the
probability of $\run(p{\downarrow}q)$
and  $\run(p{\uparrow})$, respectively. Observe that
$[p{\uparrow}] = 1 - \sum_{q \in Q} [p{\downarrow}q]$ for
every $p \in Q$.

At various places in this paper we rely on the following proposition
proven in~\cite{EWY:one-counter} (recall that
we adopt the unit-cost rational arithmetic RAM model of computation):

\begin{proposition}\label{prop:termprobs}
\label{prop-termination}
Let $\A = (Q,\delta^{=0},\delta^{>0},P^{=0},P^{>0})$ be a pOC,
and $p,q \in Q$.
\begin{itemize}
\item The problem whether $[p{\downarrow}q] >0$ is decidable in
  polynomial time.
\item If $[p{\downarrow}q] >0$, then $[p{\downarrow}q] \geq \xmin^{|Q|^3}$,
  where $\xmin$ is the least (positive) probability used in the rules
  of~$\A$.
\item The probability $[p{\downarrow}q]$ can be approximated up to an
  arbitrarily small relative error $\varepsilon > 0$
  in a time polynomial in $|\A|$ and $\log(1/\varepsilon)$.
\end{itemize}
\end{proposition}
Due to Proposition~\ref{prop-termination}, the set $\pterm$ of all pairs
$(p,q)\in Q\times Q$ satisfying $[p{\downarrow}q]>0$ is computable in
polynomial time.

%%% Local Variables:
%%% mode: latex
%%% TeX-master: "main"
%%% End:

\section{Expected Termination Time}
\label{sec-etime}

In this section we give an efficient algorithm which approximates
the expected termination time in pOC up to an arbitrarily small
relative (or even absolute) error $\varepsilon > 0$.

For the rest of this section, we fix a pOC
$\A = (Q,\delta^{=0},\delta^{>0},P^{=0},P^{>0})$. For all \mbox{$p,q \in Q$},
let $R_{p{\downarrow}q} : \run(p(1)) \rightarrow \Nseto$ be
a random variable defined as follows:
\begin{eqnarray*}
   R_{p{\downarrow}q}(w) & = &
                       \begin{cases}
                          k & \mbox{if } w \in \run(p{\downarrow}q)
                              \mbox{ and $k$ is the least index such
                                     that $w(k) = q(0)$;}\\
                          0 & \mbox{otherwise.}
                       \end{cases}\\[.5ex]
\end{eqnarray*}
If $(p,q)\in \pterm$, we use $E(p{\downarrow}q)$ to denote the
conditional expectation $\E[R_{p{\downarrow}q} \mid \run(p{\downarrow}q)]$.
Note that $E(p{\downarrow}q)$ can be finite even if $[p{\downarrow}q] < 1$.

% Our goal is to prove that for all $(p,q)\in T^{>0}$ the value
% of $E(p{\downarrow}q)$ can be efficiently approximated.

The first problem we have to deal with is that the expectation
$E(p{\downarrow}q)$ can be infinite, as illustrated by the
following example.
\begin{example}
  Consider a simple pOC with only one control state~$p$ and two positive rules
  $(p,-1,p)$ and $(p,1,p)$ that are both assigned the probability $1/2$.
  Then $[p{\downarrow}p] =1$, and due to results of
  \cite{EKM:prob-PDA-expectations},
  $E(p{\downarrow}p)$ is the least solution
  (in $\Rset^{+} \cup \{\infty\}$) of the equation
  $x = 1/2 + 1/2 (1+2x)$, which is $\infty$.
\end{example}
We proceed as follows. First, we show that the problem whether
$E(p{\downarrow}q) = \infty$ is decidable in polynomial time
(Section~\ref{subsec:inftermtime}). Then, we eliminate all infinite
expectations, and show how to approximate the finite values of
the remaining $E(p{\downarrow}q)$ up to a given absolute (and hence
also relative) error $\varepsilon > 0$ efficiently
(Section~\ref{subsec:fintermtime}).
\subsection{Finiteness of the expected termination time}
\label{subsec:inftermtime}
Our aim is to prove the following:
\begin{theorem}
\label{thm-exp-infinite}
  Let $(p,q) \in T^{>0}$. The problem whether
  $E(p{\downarrow}q)$ is finite is decidable in polynomial time.
\end{theorem}
%
%\subsubsection{Proof of Theorem~\ref{thm-exp-infinite}}
Theorem~\ref{thm-exp-infinite} is proven by analysing the underlying
finite-state Markov chain~$\X$ of the considered pOC~$\A$.
The transition matrix $A \in [0,1]^{Q \times Q}$ of $\X$ is given by
\[
   A_{p,q} = \sum_{(p,c,q) \in \delta^{>0}} P^{>0}(p,c,q).
\]
We start by assuming that $\X$ is strongly connected (i.e. that for all $p,q\in Q$ there is a path from $p$ to $q$ in $\X$). Later we show how to generalize
our results to an arbitrary $\X$.
\bigskip

\noindent{\bf Strongly connected $\X$:}
Let $\valpha \in (0,1]^Q$ be the
\emph{invariant distribution} of~$\X$, i.e., the unique (row)
vector satisfying $\valpha A = \valpha$ and $\valpha \vone = 1$ (see, e.g.,
\cite[Theorem 5.1.2]{KS60}). Further, we define the (column) vector
$\vs \in \Rset^Q$ of \emph{expected counter changes} by
\[
  \vs_p = \sum_{(p,c,q) \in \delta^{>0}} P^{>0}(p,c,q) \cdot c
\]
and the  \emph{trend} $t \in \Rset$ of $\X$ by $t = \valpha \vs$.
Note that~$t$ is easily computable in polynomial time.
Now consider some $E(p{\downarrow}q)$, where $(p,q) \in T^{>0}$. We show the
following:
\begin{itemize}
\item[(A)] If $t \neq 0$, then $E(p{\downarrow}q)$ is finite.
\item[(B)] If $t = 0$, then $E(p{\downarrow}q) = \infty$ iff the set
  $\pre^*(q(0)) \cap \post^*(p(1))$ is infinite, where
  \begin{itemize}
  \item $\pre^*(q(0))$ consists of all $r(k)$ that can reach $q(0)$
    along a run $w$ in $\M_\A$ such that the counter stays positive
    in all configurations preceding the visit to $q(0)$;
  \item $\post^*(p(1))$ consists of all $r(k)$ that can be reached from $p(1)$
    along a run $w$ in $\M_\A$ where the counter stays positive
    in all configurations preceding the visit to  $r(k)$.
  \end{itemize}
\end{itemize}
Note that the conditions of Claims~(A) and~(B) are easy to verify in
polynomial time. (Due to \cite{EHRS:MC-PDA}, there are finite-state automata
constructible in polynomial time
recognizing the sets $\pre^*(q(0))$ and $\post^*(p(1))$. Hence,
we can efficiently compute a finite-state automaton $\calF$
recognizing  the set $\pre^*(q(0)) \cap \post^*(p(1))$ and check whether the
language accepted by $\calF$ is infinite.)
Thus, if $\X$ is strongly connected and $(p,q)\in T^{>0}$,
we can decide in polynomial time whether $E(p{\downarrow}q)$ is finite.

It remains to prove Claims~(A) and~(B).
% Claim~(A) is proven by demonstrating that if $t \neq 0$, then the
% probability $[p{\downarrow}q,i]$ of all $w \in \run(p{\downarrow}q)$
% that visit $q(0)$ in \emph{exactly}~$i$ transitions decays exponentially
% in~$i$. Hence, $E(p{\downarrow}q) =
% \sum_{i=1}^{\infty} i \cdot [p{\downarrow}q,i]/[p{\downarrow}q]$ is finite.
% Claim~(B) is obtained by combining some observations about the structure
% of $\M_{\A}$ together with an argument stating that if
% $[p{\downarrow}q,i] > 0$ for infinitely many $i$'s, then the decay of
% $[p{\downarrow}q,i]$ is ``sufficiently slow'' to make the sum
% $\sum_{i=1}^{\infty} i \cdot [p{\downarrow}q,i]/[p{\downarrow}q] =
% E(p{\downarrow}q)$ infinite.
This is achieved by employing a generic observation which
connects the study of pOC to martingale theory. Recall that
a stochastic process  $\ms{0},\ms{1},\dots$ is a martingale if,
for all $i \in \Nset$, $\E(|\ms{i}|) < \infty$, and
\mbox{$\E(\ms{i+1} \mid \ms{1},\dots,\ms{i}) = \ms{i}$} almost surely. Let us fix
some initial configuration \mbox{$r(c) \in Q \times \Nset$}.
Our aim is to construct a suitable martingale over $\run(r(c))$.
Let $\ps{i}$ and $\cs{i}$ be random variables
which to every run $w \in \run(r(c))$ assign the control state
and the counter value of the configuration $w(i)$, respectively.
Note that if the vector~$\vs$ of expected counter changes
is constant, i.e., $\vs = \vone \cdot t$ where $t$ is the trend of $\X$,
then we can define a martingale $\ms{0},\ms{1},\dots$ simply by
\[
 \ms{i} = \begin{cases}
            \cs{i} \ -\  i \cdot t &
                 \text{if $\cs{j} \ge 1$ for all $0 \leq j < i$;}\\
            \ms{i-1} & \text{otherwise.}
          \end{cases}
\]
Since~$\vs$ is generally not constant, we might try to ``compensate''
the difference among the individual control states by a suitable
vector $\vv \in \Rset^{Q}$. The next proposition shows that
this is indeed possible.
\begin{proposition}
\label{prop-martingale}
  There is a vector $\vv \in \Rset^{Q}$ such that the stochastic process
  $\ms{0},\ms{1},\dots$ defined by
  \[
    \ms{i} = \begin{cases}
            \cs{i} \ +\ \vv_{\ps{i}}  \ -\  i \cdot t &
                 \text{if $\cs{j} \ge 1$ for all $0 \leq j < i$;}\\
            \ms{i-1} & \text{otherwise}
          \end{cases}
  \]
  is a martingale, where $t$ is the trend of $\X$.

  Moreover, the vector $\vv$ satisfies
  $\vmax-\vmin\  \le\  2 |Q| / \xmin^{|Q|}$, where $\xmin$ is the smallest
  positive transition probability in~$\X$, and $\vmax$ and $\vmin$
  are the maximal and the minimal components of $\vv$, respectively.
\end{proposition}
Due to Proposition~\ref{prop-martingale}, powerful results of
martingale theory such as optional stopping theorem or Azuma's inequality
(see, e.g., \cite{Rosenthal:book,Williams:book}) become applicable to pOC.
In this paper, we use the constructed martingale to complete the
proof of Claims~(A) and~(B), and to establish the crucial
\emph{divergence gap theorem} in Section~\ref{sec-LTL} (due to
space constraints, we only include brief sketches of
Propositions~\ref{lem:expected-term-bound-prob}
and~\ref{prop-term-inf} which demonstrate the use of Azuma's inequality
and optional stopping theorem).
The range of possible applications of Proposition~\ref{prop-martingale}
is of course wider.

\smallskip

\noindent
\textbf{A proof of Claim A.}
For every $i \in \Nset$, let $\run(p{\downarrow}q,i)$ be the set of all
$w \in \run(p{\downarrow}q)$ that visit $q(0)$ in \emph{exactly}~$i$
transitions, and let $[p{\downarrow}q,i]$ be the probability of
$\run(p{\downarrow}q,i)$.  Claim~(A) is proven by demonstrating that
if $t \neq 0$, then the probabilities $[p{\downarrow}q,i]$ decay
exponentially in~$i$. Hence, $E(p{\downarrow}q) =
\sum_{i=1}^{\infty} i \cdot [p{\downarrow}q,i]/[p{\downarrow}q]$ is finite.
%
%\TOMAS 24.1. 14:00
%The exponential decay of $[p{\downarrow}q,i]$ is a direct corollary to
%our next proposition, which proves a slightly more general result that
%becomes useful later.
%%
%\begin{proposition} \label{lem:expected-term-bound-prob}
% Let $p(k)$ be an initial configuration, and let $H_i$ be the set of all
% runs initiated in $p(k)$ that visit a configuration with zero counter
% in exactly $i$~transitions.
% % i.e., $H_i := \{w \in \run(p(k)) \mid \cs{i} = 0 \ \land \ \forall 0 \le j < i: \cs{j} \ge 1\}$.
% Let
%  \[
%   a = \exp\left(- \frac{t^2}{8 (\va + t + 1)^2} \right)\,.
%  \]
% Note that $0 < a < 1$. Further, let
%  \[
%   h = \begin{cases} 2 \cdot \frac{-\va - \cs{0}}{t} & \text{if \ $t < 0$} \\
%                      2 \cdot \frac{ \va - \cs{0}}{t} & \text{if \ $t > 0$ .}
%        \end{cases}
%  \]
% Then for all $i \in \Nset$ with $i \ge h$ we have that $\calP(H_i) \le a^i$.
%\end{proposition}
\begin{proposition} \label{lem:expected-term-bound-prob}
There are $0<a<1$ and $h\in \Nset$ such that for all $i\geq h$ we have that $[p{\downarrow}q,i]\leq a^i$.
\end{proposition}
 \begin{proof}[Sketch]
  Consider the martingale $\ms{0},\ms{1},\ldots$ over $\run(p(1))$ as defined in Proposition~\ref{prop-martingale}.
  A relatively straightforward computation reveals that for sufficiently large $h\in \Nset$ and all $i\geq h$ we have the following:
  If $t<0$, then $[p{\downarrow}q,i]\le \calP\left(\ms{i} - \ms{0} \ge (i/2) \cdot (-t)\right)$, and
  if $t>0$, then $[p{\downarrow}q,i]\le \calP\left(\ms{0} - \ms{i} \ge (i/2) \cdot t \right)$.
  In each step, the martingale value changes by at most $\vmax - \vmin + t + 1$, where $\vv$ is from Proposition~\ref{prop-martingale}.
  Hence Azuma's inequality (see~\cite{Williams:book}) asserts for $t \ne 0$ and $i \ge h$:
  \begin{align*}
   [p{\downarrow}q,i]
   & \quad \le \quad \exp  \left(- \frac{(i/2)^2 t^2}{2 i (\vmax - \vmin + t + 1)^2} \right) && \text{(Azuma's inequality)} \\
   & \quad =\quad a^i \,.
  \end{align*}
   Here $a=\exp\left(- t^2\,/\,\, 8 (\vmax - \vmin + t + 1)^2 \right)$.
 \qed
 \end{proof}
It follows directly from Proposition~\ref{lem:expected-term-bound-prob} that
\[
E(p{\downarrow}q) \quad = \quad
\sum_{i=1}^{\infty} i \cdot
   \frac{[p{\downarrow}q,i]}{[p{\downarrow}q]} \quad \leq \quad
\frac{1}{[p{\downarrow}q]}\left( \sum_{i=1}^{h-1} i \cdot [p{\downarrow}q,i]+\sum_{i=h}^{\infty} i\cdot a^i\right) \quad < \quad \infty
\]

%%% Local Variables:
%%% mode: latex
%%% TeX-master: "main"
%%% End:

\smallskip

\noindent
\textbf{A proof of Claim B.}
%We start by introducing some notation. For every configuration
%$p(i)$, where $i \geq 1$, and every state $q$, let
%$\run(p(i){\downarrow}q)$ be the set of all $w \in \run_\A(p(i))$
%that visit $q(0)$, and the counter stays above~$0$
%in all configurations preceding this visit. Further, we define a
%random variable $R_{p(i){\downarrow}q} : \run(p(i)) \rightarrow \Nset$
%as follows:
%\begin{eqnarray*}
%   R_{p(i){\downarrow}q}(w) & = &
%                       \begin{cases}
%                          k & \mbox{if } w \in \run(p(i){\downarrow}q)
%                              \mbox{ and $k$ is the least index such
%                                     that $w(k) = q$;}\\
%                          0 & \mbox{otherwise.}
%                       \end{cases}\\[.5ex]
%\end{eqnarray*}
%We also put $\run(p(i){\downarrow}j) = \bigcup_{q\in Q} \run(p(i){\downarrow}q)$
%and define a random variable
%\mbox{$R_{p(i){\downarrow}} : \run(p(i)) \rightarrow \Nset$}
%by
%\begin{eqnarray*}
%   R_{p(i){\downarrow}}(w) & = &
%                       \begin{cases}
%                          k & \mbox{if } w \in \run(p(i){\downarrow})
%                              \mbox{ and $k$ is the least index s.t.{}
%                                  the counter is $0$ in $w(k)$;}\\
%                          0 & \mbox{otherwise.}
%                       \end{cases}\\[.5ex]
%\end{eqnarray*}
%The conditional expectations
%$\E[R_{p(i){\downarrow}q} \mid \run(p(i){\downarrow}q)]$
%and
%$\E[R_{p(i){\downarrow}} \mid \run(p(i){\downarrow})]$ are denoted by
%$E(p(i){\downarrow}q)$ and $E(p(i){\downarrow})$, respectively.
%
We start with the ``$\Rightarrow$'' direction of Claim~(B), which
is easy to prove by contradiction. Intuitively, if
$\pre^*(q(0)) \cap \post^*(p(1))$
is finite, then we can transform the states of
$\pre^*(q(0)) \cap \post^*(p(1))$ into a finite-state Markov chain
and show that $E(p{\downarrow}q)$ is finite.
\begin{proposition}
  If $\pre^*(q(0)) \cap \post^*(p(1))$ is finite, then
  $E(p{\downarrow}q)$ is also finite.
\end{proposition}
The other direction of Claim~(B) is more complicated. Let us first
introduce some notation. For every $k \in \Nseto$, let
$Q(k)$ be the set of all configurations where the counter
value equals~$k$. Let $p,q \in Q$ and $\ell,k \in \Nseto$, where
$\ell > k$. An \emph{honest path from $p(\ell)$ to $q(k)$}
is a finite path $w$ from $p(\ell)$ to $q(k)$ such that
the counter stays above~$k$ in all configurations of $w$ except for
the last one. We use $\mathit{hpath}(p(\ell),Q(k))$ to denote the set
of all honest paths from $p(\ell)$ to some $q(k)\in Q(k)$.
For a given $P \subseteq \mathit{hpath}(p(\ell),Q(k))$, the
\emph{expected lenght of an honest path in $P$} is defined as
$\sum_{w \in P} \calP(\run(w))\cdot \len(w)$.
Using the above constructed  martingale, we show the following:
\begin{proposition}
\label{prop-term-inf}
  If $\pre^*(q(0))$ is infinite, then almost all runs initiated
  in an arbitrary configuration reach $Q(0)$. Moreover, there is
  $k_1\in \Nset$ such that, for all $\ell\geq k_1$, the expected
  length of an honest path from $r(\ell)$ to $Q(0)$ is infinite.
%
%  If $\pre^*(q(0))$, then there is a $k_1 \in \Nset$ such that for
%  every configuration $r(\ell)$ where $\ell \geq k_1$ we have that
%  $\calP(\run(r(\ell){\downarrow})) = 1$ and
%  $E(r(\ell){\downarrow}) = \infty$.
\end{proposition}
\begin{proof}[Sketch]
Assume that $\pre^*(q(0))$ is infinite. The fact that almost all runs initiated
  in an arbitrary configuration reach $Q(0)$ follows from results of~\cite{BBEKW:OC-MDP}.

Consider an initial configuration $r(\ell)$ with $\ell + \vv_r > \vmax$.
We will show that the expected length of an honest path from~$r(\ell)$ to~$Q(0)$ is infinite; i.e., we can take $k_1 := \lceil \vmax - \vmin + 1\rceil$.
Consider the martingale $\ms{0},\ms{1},\dots$ defined in Proposition~\ref{prop-martingale} over $\run(r(\ell))$.
Note that as $t=0$, the term $i\cdot t$ vanishes from the definition of the martingale.

Now let us fix $k\in \Nset$ such that $\ell + \vv_r < \vmax + k$ and
define a {\em stopping time}~$\tau$ (see e.g.~\cite{Williams:book})
which returns the first point in time in which either $m^{(\tau)}\geq \vmax + k$, or $m^{(\tau)}\leq \vmax$.
A routine application of optional stopping theorem gives us the following
\begin{equation}
\calP(\ms\tau \geq \vmax + k)\quad\geq\quad \frac{\ell + \vv_r - \vmax}{k+M} \,. \label{eq:mart-opt-stop-upper-sketch}
\end{equation}
Denote by~$T$ the number of steps to hit~$Q(0)$.
Note that $\ms\tau \ge \vmax + k$ implies
 $
  \cs\tau = \ms\tau - \vv_{\ps\tau} \ge \vmax + k - \vv_{\ps\tau} \ge k,
 $
 and thus also $T \ge k$, as at least $k$ steps are required to decrease the counter value from~$k$ to~$0$.
It follows that
$\calP(\ms\tau \ge \vmax + k)\le \calP(T \ge k)$.
By putting this inequality together with the inequality~\eqref{eq:mart-opt-stop-upper-sketch} we obtain
\[
 \E T\ =\ \sum_{k \in \Nset} \calP(T \ge k)\ \ge\ \sum_{k=\ell + 1}^\infty \calP(T \ge k)
     \ \ge\ \sum_{k=\ell + 1}^\infty \frac{\ell + \vv_r - \vmax}{k+M}\ =\ \infty \,.
\]
\qed
\end{proof}

\noindent
Further, we need the following observation about the structure of
$\M_\A$, which holds also for non-probabilistic one-counter automata:
\begin{proposition}
\label{prop-pre-closed}
  There is $k_2 \in \Nset$ such that for every configuration
  $r(\ell) \in \pre^*(q(0))$, where $\ell \geq k_2$, we have that
  if $r(\ell) \tran{} r'(\ell')$, then $r'(\ell') \in \pre^*(q(0))$.
\end{proposition}
To show that  $E(p{\downarrow}q) = \infty$, it suffices
to identify a subset \mbox{$W \subseteq R(p{\downarrow}q)$} such that
\mbox{$\calP(W) > 0$} and $\E[R_{p{\downarrow}q} \mid W] = \infty$.
Now observe that if $\pre^*(q(0)) \cap \post^*(p(1))$ is infinite,
there is a configuration $r(\ell) \in \pre^*(q(0))$ reachable
from $p(1)$ along a finite path $u$ such that $\ell \geq k_1+k_2$, where $k_1$ and $k_2$ are the constants of
Propositions~\ref{prop-term-inf}
and~\ref{prop-pre-closed}.

Due to Proposition~\ref{prop-term-inf}, the expected length of an honest
path from $r(\ell-k_2)$ to $Q(0)$ is infinite. Howeover,
then also the expected length of an honest path from
$r(\ell)$ to $Q(k_2)$ is infinite. This means that
there is a state $s \in Q$ such that the expected length of an honest
path from $r(\ell)$ to $s(k_2)$ in infinite.
Further, it follows directly
from Proposition~\ref{prop-pre-closed} that $s(k_2) \in \pre^*(q(0))$
because there is an honest path from $r(\ell)$ to $s(k_2)$.

Now consider the set $W$ of all runs $w$ initiated in $p(1)$
that start with the finite path~$u$, then follow an honest path from
$r(\ell)$ to $s(k_2)$, and then follow
an honest path from $s(k_2)$ to $q(0)$. Obviously,
$\calP(W) > 0$, and $\E[R_{p{\downarrow}q} \mid W] = \infty$ because
the expected length of the middle subpath is infinite. Hence,
$E(p{\downarrow}q) = \infty$ as needed.

%%% Local Variables:
%%% mode: latex
%%% TeX-master: "main"
%%% End:

\bigskip

\noindent {\bf Non-strongly connected $\X$:} The general case still requires
some extra care. First, realize that each BSCC $\B$ of $\X$ can be seen
as a strongly connected finite-state Markov chain, and hence all
notions and arguments of the previous subsection can be applied to~$\B$
immediately (in particular, we can compute the trend of $\B$ in polynomial
time). We prove the following claims:
\begin{itemize}
\item[(C)] If $q$ does not belong to a BSCC of $\X$, then
  $E(p{\downarrow}q)$ is finite.
\item[(D)] If $q$ belongs to a BSCC $\B$ of $\X$ such that the trend
  of $\B$ is different from~$0$, then $E(p{\downarrow}q)$ is finite.
\item[(E)] If $q$ belongs to a BSCC $\B$ of $\X$ such that the trend
  of $\B$ is~$0$, then $E(p{\downarrow}q) = \infty$ iff
  the set $\pre^*(q(0)) \cap \post^*(p(1))$ is infinite.
\end{itemize}
Note that the conditions of Claims~(C)-(E) are verifiable in polynomial
time.

Intuitively, Claim~(C) is proven by observing that if
$q$ does not belong to a BSCC of $\X$, then for all
$s(\ell) \in \post^*(p(1))$, where $\ell \geq |Q|$, we have that
$s(\ell)$ can reach a configuration outside $\pre^*(q(0))$
in at most $|Q|$ transitions. It follows that the probability of
performing an honest path from $p(1)$ to $q(0)$ of length~$i$
decays exponentially in~$i$, and hence $\E(p{\downarrow}q)$ is finite.

Claim~(D) is obtained by combining the arguments of Claim~(A) together
with the fact that the conditional expected number of transitions needed
to reach $\B$ from $p(0)$, under the condition that $\B$ is indeed reached
from $p(0)$, is finite (this is a standard result for finite-state
Markov chains).

Finally, Claim~(E) follows by re-using the arguments of Claim~(B).

%, but in the end we obtain the
%following:
%
%A detailed proof of Theorem~\ref{thm-exp-infinite} is given in Appendix.
%
%
\subsection{Efficient approximation of finite expected termination time}\label{subsec:fintermtime}
Let us denote by $\pfterm$ the set of all pairs $(p,q)\in \pterm$
satisfying $E(p{\downarrow}q)<\infty$. Our aim is to prove the following:
\begin{theorem}
\label{thm-regular}
  For all $(p,q) \in \pfterm$, the value of $E(p{\downarrow}q)$ can
  be approximated up to an arbitrarily small absolute error $\varepsilon > 0$
  in time polynomial in $|\A|$ and $\log(1/\varepsilon)$.
\end{theorem}
Note that if $y$ approximates $E(p{\downarrow}q)$ up to an absolute
error  $1 > \varepsilon > 0$, then $y$ approximates $E(p{\downarrow}q)$
also up to the relative error $\varepsilon$ because
$E(p{\downarrow}q) \geq 1$.

The proof of Theorem~\ref{thm-regular} is based on the fact that
the vector of all
$E(p{\downarrow}q)$, where $(p,q)\in \pfterm$, is the unique solution
of a system of linear equations whose coefficients can be efficiently
approximated (see below). Hence, it suffices to approximate the
coefficients, solve the approximated equations, and then bound the
error of the approximation using standard arguments from numerical
analysis.

Let us start by setting up the system of linear equations for
$E(p{\downarrow}q)$.
For all $p,q \in T^{>0}$, we fix a fresh variable $V(p{\downarrow}q)$,
and construct the following system of linear equations, $\calL$, where the
termination probabilities are treated as constants:
%\stefan{What are symbolic constants?}
\begin{eqnarray*}
%\label{eqn-term}
  V(p{\downarrow}q) & = & \sum_{(p,-1,q) \in \delta^{>0}}
    \frac{P^{>0}(p,-1,q)}{[p{\downarrow}q]}
  \ + \
  \sum_{(p,0,t) \in \delta^{>0}}
      \frac{P^{>0}(p,0,t) \cdot [t{\downarrow}q]}{[p{\downarrow}q]} \cdot
      \bigg(1+ V(t{\downarrow}q)\bigg)\\[1ex]
  & + &
  \sum_{(p,1,t) \in \delta^{>0}} \sum_{r \in Q}
      \frac{P^{>0}(p,1,t) \cdot
           [t{\downarrow}r]\cdot[r{\downarrow}q]}{[p{\downarrow}q]} \cdot
      \bigg(1+ V(t{\downarrow}r)+ V(r{\downarrow}q)\bigg)
\end{eqnarray*}
It has been shown in \cite{EKM:prob-PDA-expectations}
that the tuple of all $E(p{\downarrow}q)$,
where $(p,q)\in T^{>0}$, is the least solution
of~$\calL$ in $\Rset^{+} \cup \{\infty\}$
with respect to component-wise ordering (where $\infty$ is treated
according to the standard conventions).
Due to Theorem~\ref{thm-exp-infinite}, we can further simplify the
system~$\calL$ by erasing the defining equations for
all $V(p{\downarrow}q)$ such that $E(p{\downarrow}q) = \infty$
(note that if $E(p{\downarrow}q) < \infty$, then the defining equation
for $V(p{\downarrow}q)$ in~$\calL$ cannot contain any
variable $V(r{\downarrow}t)$ such that $E(r{\downarrow}t) = \infty$).

Thus, we obtain the system~$\calL'$. It is straightforward
to show that the vector of all finite $E(p{\downarrow}q)$ is
the \emph{unique} solution of the system $\calL'$ (see, e.g., Lemma~6.2.3 and
Lemma~6.2.4 in~\cite{Brazdil:PhD}).
%First, let us realize that the matrix $H$ is regular. Suppose the converse,
%i.e., $\calL'$ has another solution $\vec{F}$ in~$\Rset$.
%\stefan{The uniqueness argument is nice and understandable,
%on the other hand it's known and we could refer to Tomas' thesis or some old paper.}
%Then we can pick a suitable $a \in (-1,1)$ such that the absolute
%value of each component in  $a(\vec{E} - \vec{F})$ is less than~$1$,
%and at least one component of this vector is negative.
%Hence, $a(\vec{E} - \vec{F}) + \vec{E}$ is another solution of~$\calL'$
%in $\Rset^{+}$, which is either strictly less than $\vec{E}$ or
%incomparable with $\vec{E}$. This contradicts the fact that
%$\vec{E}$ is the least solution of~$\calL'$ in $\Rset^{+}$.
If we rewrite~$\calL'$ into a standard matrix form, we obtain a
system $\vec{V} = H \cdot \vec{V} + \vec{b}$, where $H$ is a
nonsingular nonnegative matrix, $\vec{V}$ is the vector of variables
in~$\calL'$, and $\vec{b}$ is a vector.  Further,
we have that $\vec{b} = \vone$, i.e., the constant coefficients are all~$1$.
This follows from the following equality
(see \cite{EKM:prob-PDA-PCTL,EY:RMC-SG-equations}):
\begin{equation} \label{eq:termination-probabilities}
\begin{aligned} \
[p{\downarrow}q]=\sum_{(p,-1,q) \in \delta^{>0}}
    P^{>0}(p,-1,q)\ & +  \sum_{(p,0,t) \in \delta^{>0}}
     P^{>0}(p,0,t) \cdot [t{\downarrow}q] \\
     & +
  \sum_{(p,1,t) \in \delta^{>0}} \sum_{r \in Q}
      P^{>0}(p,1,t) \cdot
           [t{\downarrow}r]\cdot[r{\downarrow}q]
\end{aligned}
\end{equation}
Hence, $\calL'$ takes the form $\vec{V} = H \cdot \vec{V} + \vone$.
Unfortunately, the entries of~$H$ can take irrational values and
cannot be computed precisely in general.  However, they can be
approximated up to an arbitrarily small relative error using
Proposition~\ref{prop:termprobs}.  Denote by $G$ an approximated
version of~$H$.  We aim at bounding the error of the solution of the
``perturbed'' system $\vec{V} = G \cdot \vec{V} + \vone$ in terms of
the error of~$G$.  To measure these errors, we use the $l_\infty$ norm
of vectors and matrices, defined as follows: For a vector $\vec{V}$ we have
that $\norm{\vec{V}}=\max_i |\vec{V}_i|$, and for a matrix~$M$ we have
$\norm{M}=\max_i \sum_j |M_{ij}|$. Hence, $\norm{M} = \norm{M \cdot \vone}$
if $M$ is nonnegative.  We show the following:
%\begin{proposition}
%\label{prop-exp-approx}
%  For each $\varepsilon$, where $0 < \varepsilon < 1$, let
%  $\delta = \varepsilon / (12 c^{2|\A|})$.
%  If $\norm{G-H} \le \delta$, then the perturbed system
%   $\vec{V} = G \cdot \vec{V} + \vone$
%  has a unique solution $\vec{F}$.
%  Moreover, we have that
%  \[
%   |E(p{\downarrow} q) - \vec{F}_{pq}| \quad \leq\quad \varepsilon \qquad \text{for all $(p,q) \in \pfterm$.}
%  \]
%  Here $\vec{F}_{pq}$ is the component of $\vec{F}$ corresponding to the variable $V(p{\downarrow}q)$.
%\end{proposition}
%
%
\begin{proposition}
\label{prop-exp-approx}
  %There is a polynomial $P(x)$ such that the following holds.
  %Let $b\in \Rset^+$ satisfy $E(p{\downarrow}q)\leq b$ for all $(p,q) \in \pfterm$.
%  Let $b:=\max\left\{E(p{\downarrow}q)\mid (p,q)\in \pfterm\right\}$. Then $b$ is bounded from above by $P(|\A|)$, where
%  $P$ is a fixed polynomial.
  Let $b \ge \max\left\{E(p{\downarrow}q)\mid (p,q)\in \pfterm\right\}$.
  Then for each $\varepsilon$, where $0 < \varepsilon < 1$, let
  $\delta = \varepsilon\, / (12\cdot b^2)$.
  If $\norm{G-H} \le \delta$, then the perturbed system
   $\vec{V} = G \cdot \vec{V} + \vone$
  has a unique solution $\vec{F}$, and in addition, we have that
  \[
   |E(p{\downarrow} q) - \vec{F}_{pq}| \quad \leq\quad \varepsilon \qquad \text{for all $(p,q) \in \pfterm$.}
  \]
  Here $\vec{F}_{pq}$ is the component of $\vec{F}$ corresponding to the variable $V(p{\downarrow}q)$.
\end{proposition}
%
%
%\label{prop-exp-approx}
%  %There is a polynomial $P(x)$ such that the following holds.
%  Let $b\in \Rset^+$ satisfy $E(p{\downarrow}q)\leq b$ for all $(p,q) \in \pfterm$.
%  For each $\varepsilon$, where $0 < \varepsilon < 1$, let
%  $\delta = \varepsilon\, / \left(2^{2\cdot P(|\A|)}\right)$.
%  If $\norm{G-H} \le \delta$, then the perturbed system
%   $\vec{V} = G \cdot \vec{V} + \vone$
%  has a unique solution $\vec{F}$.
%  Moreover, we have that
%  \[
%   |E(p{\downarrow} q) - \vec{F}_{pq}| \quad \leq\quad \varepsilon \qquad \text{for all $(p,q) \in \pfterm$.}
%  \]
%  Here $\vec{F}_{pq}$ is the component of $\vec{F}$ corresponding to the variable $V(p{\downarrow}q)$.
%\end{proposition}
The proof of Proposition~\ref{prop-exp-approx} is based on estimating
the size of the condition number $\kappa = \norm{1-H} \cdot
\norm{(1-H)^{-1}}$ and applying standard results of numerical
analysis. The $b$ in Proposition~\ref{prop-exp-approx} can be estimated
as follows:
\newcommand{\stmtpropexptimeboundspecial}{
 Let $\xmin$ denote the smallest nonzero probability in~$A$.
 Then we have:
 \[
  E(p{\downarrow}q) \quad\le\quad 85000 \cdot |Q|^6 / \left( \xmin^{6 |Q|^3} \cdot t_{\min}^4 \right) \qquad \text{for all $(p,q) \in \pfterm$,}
 \]
 where $t_{\min} = \{|t| \ne 0 \mid \text{$t$ is the trend in a BSCC of~$\X$}\}$.
}
\begin{proposition} \label{prop:exp-time-bound-special}
 \stmtpropexptimeboundspecial
\end{proposition}
Although $b$ appears large, it is really the value
of $\log(1/b)$ which matters, and it is still reasonable.
Theorem~\ref{thm-regular} now follows by combining
Propositions~\ref{prop:exp-time-bound-special}, \ref{prop-exp-approx} and~\ref{prop:termprobs},
because the approximated matrix~$G$ can be computed using a number of arithmetical operations which is
polynomial in $|\A|$ and $\log(1/\varepsilon)$.

\section{Quantitative Model-Checking of  $\omega$-regular
Properties}
\label{sec-LTL}

In this section, we show that for every $\omega$-regular property
encoded by a deterministic Rabin automaton, the probability
of all runs in a given pOC that satisfy the property can be approximated
up to an arbitrarily small relative error $\varepsilon>0$ in
polynomial time. This is achieved by designing and analyzing
a new quantitative model-checking algorithm for pOC
and $\omega$-regular properties, which
is \emph{not} based on techniques developed for pPDA and
RMC in \cite{EKM:prob-PDA-PCTL,EY:RMC-LTL-complexity,EY:RMC-LTL-QUEST}.

Recall that a deterministic Rabin automaton (DRA) over a finite
alphabet $\Sigma$ is a deterministic finite-state automaton
$\R$ with total transition function and \emph{Rabin acceptance
condition} $(E_1,F_1),\ldots,(E_k,F_k)$, where $k \in \Nset$, and all
$E_i$, $F_i$ are subsets of control states of~$\R$.
For a given infinite word $w$ over $\Sigma$, let $\inf(w)$ be the
set of all control states visited infinitely often along the unique
run of $\R$ on $w$. The word $w$ is accepted by $\R$ if there is
$i \leq k$ such that $\inf(w) \cap E_i = \emptyset$ and
$\inf(w) \cap F_i \neq \emptyset$.

Let $\Sigma$ be a finite alphabet, $\R$ a DRA over $\Sigma$,
and $\A = (Q,\delta^{=0},\delta^{>0},P^{=0},P^{>0})$ a pOC. A
\emph{valuation} is a function $\nu$ which to every configuration
$p(i)$ of $\A$ assigns a unique letter of~$\Sigma$. For simplicity, we
assume that $\nu(p(i))$ depends only on the control state $p$ and the
information whether $i \geq 1$. Intuitively, the letters
of $\Sigma$ correspond to collections of predicates that are
valid in a given configuration of $\A$. Thus, every run
$w \in \run_{\A}(p(i))$ determines a unique infinite word
$\nu(w)$ over $\Sigma$ which is either accepted by $\R$ or not.
The main result of this section is the following theorem:

\begin{theorem}
\label{thm-omega-regular}
  For every $p \in Q$, the probability of all $w \in \run_\A(p(0))$
  such that $\nu(w)$ is accepted by $\R$ can be approximated
  up to an arbitrarily small relative error $\varepsilon > 0$
  in time polynomial in $|\A|$, $|\R|$, and
  $\log(1/\varepsilon)$.
\end{theorem}
Our proof of Theorem~\ref{thm-omega-regular} consists of three steps:
\begin{itemize}
\item[1.] We show that the problem of our interest is equivalent
  to the problem of computing the probability of all accepting
  runs in a given pOC $\A$ with Rabin acceptance condition.
\item[2.] We introduce a finite-state Markov chain $\G$
  (with possibly irrational transition probabilities) such that
  the probability of all accepting runs in $\M_\A$ is
  equal to the probability of reaching a ``good'' BSCC in $\G$.
\item[3.] We show how to compute the probability of reaching
  a ``good'' BSCC in $\G$ with relative error at most~$\varepsilon$
  in time polynomial in $|\A|$ and $\log(1/\varepsilon)$.
\end{itemize}
Let us note that Steps~1 and~2 are relatively simple, but
Step~3 requires several insights. In particular, we cannot
solve Step~3 without bounding a positive non-termination probability in
pOC (i.e., a positive probability of the form $[p{\uparrow}]$) away from
zero. This is achieved in our ``divergence gap theorem'' (i.e.,
Theorem~\ref{thm-gap}), which is based on applying Azuma's inequality
to the martingale constructed in Section~\ref{sec-etime}. Now
we elaborate the three steps in more detail.
\smallskip

\noindent
\textbf{Step~1.}
For the rest of this section, we fix a pOC
$\A = (Q,\delta^{=0},\delta^{>0},P^{=0},P^{>0})$, and a \emph{Rabin acceptance
condition} $(\calE_1,\F_1),\ldots,(\calE_k,\F_k)$, where $k \in \Nset$
and $\calE_i,\F_i \subseteq Q$ for all \mbox{$1 \leq i \leq k$}.
For every run $w \in \run_{\A}$, let $\inf(w)$ be the set of
all $p \in Q$ visited infinitely often along~$w$. We use
$\run_\A(p(0),\acc)$ to denote the set of all \emph{accepting runs}
$w \in \run_{\A}(p(0))$ such that  $\inf(w) \cap \calE_i = \emptyset$ and
$\inf(w) \cap \F_i \neq \emptyset$ for some $i \leq k$.
Sometimes we also write $\run_\A(p(0),\rej)$ to denote the set
  $\run_{\A}(p(0)) \smallsetminus \run_\A(p(0),\acc)$ of \emph{rejecting}
runs.

Our next proposition says that the problem of computing/approximating
the probability of all runs $w$ in a given pOC that are accepted by
a given DRA is efficiently reducible to the problem of
computing/approximating the probability of all accepting runs in a given
pOC with Rabin acceptance condition. The proof is very simple (we just
``synchronize'' a given pOC with a given DRA, and setup the Rabin
acceptance condition accordingly).

\begin{proposition}
\label{prop-product}
  Let $\Sigma$ be a finite alphabet, $\A$ a pOC, $\nu$ a valuation,
  $\R$ a DRA over $\Sigma$, and $p(0)$ a configuration of $\A$.
  Then there is a pOC $\A'$ with Rabin acceptance condition
  and a configuration $p'(0)$ of $\A'$ constructible in polynomial time
  such that the probability of all $w \in \run_{\A}(p(0))$ where $\nu(w)$
  is accepted by $\R$ is equal to the probability of all accepting
  $w \in \run_{\A'}(p'(0))$.
\end{proposition}

\smallskip

\noindent
\textbf{Step~2.} Let $\G$ be a finite-state Markov chain, where
$Q \times \{0,1\} \ \cup\ \{\acc,\rej\}$ is the set of
states (the elements of $Q \times \{0,1\}$ are written
as $p(i)$, where $i \in \{0,1\}$), and
the transitions of $\G$ are determined as follows:
  \begin{itemize}
  \item $p(0) \tran{x} q(j)$ is a transition of $\G$
     iff $p(0) \tran{x} q(j)$ is a transition of $\M_{\A}$;
  \item $p(1) \tran{x} q(0)$ iff
     $x = [p{\downarrow}q] > 0$;
  \item $p(1) \tran{x} \acc$ iff
     $x = \calP(\run_{\A}(p(1),\acc) \cap \run_\A(p{\uparrow})) > 0$;
  \item $p(1) \tran{x} \rej$ iff
     $x = \calP(\run_{\A}(p(1),\rej) \cap \run_\A(p{\uparrow})) > 0$;
  \item $\acc \tran{1} \acc$, $\rej \tran{1} \rej$;
  \item there are no other transitions.
  \end{itemize}
A BSCC $B$ of $\G$ is \emph{good} if either $B = \{\acc\}$, or
there is some $i\leq k$ such that $\calE_i \cap Q(B) = \emptyset$
and $\F_i \cap Q(B) \neq \emptyset$, where
$Q(B) = \{p \in Q \mid p(j) \in B \text{ for some } j \in \{0,1\}\}$.
For every $p \in Q$, let $\run_\G(p(0),\mathit{good})$ be the set
of all $w \in \run_{\G}(p(0))$ that visit a good BSCC of~$\G$.
The next proposition is obtained by a simple case analysis of
accepting runs in~$\M_\A$.
\begin{proposition}
\label{prop-X}
  For every $p \in Q$ we have
  $\calP(\run_\A(p(0),\acc)) = \calP(\run_\G(p(0),\mathit{good}))$.
\end{proposition}
\smallskip

\noindent
\textbf{Step~3.} Due to Proposition~\ref{prop-X}, the problem of our
interest reduces to the problem of approximating the probability of
visiting a good BSCC in the finite-state Markov chain~$\G$.
% A proof of the following proposition can be found
% in~\cite{EWY:one-counter}.
%
% {\bf THE PROP IS NEEDED IN SECTION 3, I MOVED IT TO SECTION 2}
% \begin{proposition}
% Let $\A = (Q,\delta^{=0},\delta^{>0},P^{=0},P^{>0})$ be a pOC,
% and $p,q \in Q$.
% \begin{itemize}
% \item The problem whether $[p{\downarrow}q] >0$ is decidable in
%   polynomial time.
% \item If $[p{\downarrow}q] >0$, then $[p{\downarrow}q] \geq p_{\min}^{|Q|^3}$,
%   where $p_{\min}$ is the least (positive) probability used in the rules
%   of~$\A$.
% \item The probability $[p{\downarrow}q]$ can be approximated up to an
%   arbitrarily small relative precision in polynomial time.
% \end{itemize}
% \end{proposition}
%
Since the termination probabilities in~$\A$ can be approximated efficiently
(see Proposition~\ref{prop-termination}), the main problem with
$\G$ is approximating the probabilities $x$ and $y$ in transitions
of the form $p(1) \tran{x} \acc$ and $p(1) \tran{y} \rej$.
Recall that $x$ and $y$ are the probabilities
of all $w \in \run_{\A}(p{\uparrow})$ that are accepting and rejecting,
respectively. A crucial observation is that
almost all $w \in \run_{\A}(p{\uparrow})$ still behave
accordingly with the underlying finite-state Markov chain $\X$ of $\A$
(see Section~\ref{sec-etime}). More precisely, we have the following:
\begin{proposition}
\label{prop-X-BSCC}
  Let $p \in Q$. For almost all
  $w \in \run_\A(p{\uparrow})$ we have that $w$ visits a BSCC $B$
  of $\X$ after finitely many transitions, and then it
  visits all states of $B$ infinitely often.
\end{proposition}
A BSCC $B$ of $\X$ is \emph{consistent} with the considered Rabin
acceptance condition if there is $i \leq k$ such that
$B \cap \calE_i = \emptyset$ and $B \cap \F_i \neq \emptyset$.
If $B$ is not consistent, it is \emph{inconsistent}.
An immediate corollary to Proposition~\ref{prop-X-BSCC} is the following:
\begin{corollary}
\label{cor-nonterm}
Let $\run_\A(p(1),\cons)$ and
$\run_\A(p(1),\inco)$ be the sets of all
$w \in \run_{\A}(p(1))$ such that $w$ visit a control state
of some consistent and inconsistent BSCC of $\X$, respectively.
Then
\begin{itemize}
\item $\calP(\run_\A(p(1),\acc) \cap \run_\A(p{\uparrow})) \ = \
  \calP(\run_\A(p(1),\cons) \cap \run_\A(p{\uparrow}))$
\item $\calP(\run_\A(p(1),\rej) \cap \run(p{\uparrow})) \ = \
  \calP(\run_\A(p(1),\inco) \cap \run_\A(p{\uparrow}))$
\end{itemize}
\end{corollary}
Due to Corollary~\ref{cor-nonterm}, we can reduce the problem of computing
the probabilities of transitions of the form $p(1) \tran{x} \acc$ and
$p(1) \tran{y} \rej$ to the problem of computing the probability
of non-termination in pOC. More precisely, we construct
pOC's $\A_\cons$ and $\A_\inco$ which are the same as $\A$, except that for
each control state $q$ of an inconsistent (or consistent, resp.)
BSCC of $\X$, all positive  outgoing rules of $q$ are replaced
with $q \prule{1,-1} q$.
Then $x = \calP(\run_{\A_\cons}(p{\uparrow}))$ and
$y = \calP(\run_{\A_\inco}(p{\uparrow}))$.

Due to \cite{BBEKW:OC-MDP}, the problem whether a given non-termination
probability is positive (in a given pOC) is decidable in polynomial
time. This means that the underlying graph of $\G$ is computable in
polynomial time, and hence the sets $G_0$ and $G_1$ consisting
of all states $s$ of $\G$ such that $\calP(\run_\G(s,\mathit{good}))$
is equal to $0$ and $1$, respectively, are constructible in
polynomial time. Let $G$ be the set of all states of $\G$ that are
not contained in $G_0 \cup G_1$, and let $X_\G$ be the stochastic matrix
of~$\G$. For every $s \in G$ we fix a fresh variable $V_s$ and the equation
\begin{equation*}
  V_s \quad = \quad \sum_{s'\in G} X_\G(s,s') \cdot V_{s'}
      \quad + \quad \sum_{s'\in G_1}  X_\G(s,s')
\end{equation*}
Thus, we obtain a system of linear equations
$\vec{V} = A \vec{V} + \vec{b}$ whose unique solution $\vec{V}^*$ in
$\Rset$ is the vector of probabilities of reaching a good BSCC from the
states of~$G$. This system can also be written
as $(I-A)\vec{V} = \vec{b}$. Since the elements of $A$ and $\vec{b}$
correspond to (sums of) transition probabilities in $\G$, it suffices
to compute the transition probabilities of $\G$ with a sufficiently small
relative error so that the approximate $A$ and $\vec{b}$ produce
an approximate solution where the relative error of each component is bounded
by the $\varepsilon$. By combining standard results for finite-state
Markov chains with techniques of numerical analysis, we show the
following:
\begin{proposition}
\label{prop-visiting-approx}
  Let $c = 2|Q|$. For every $s \in G$, let $R_s$ be the probability
  of visiting a BSCC of $\G$ from $s$ in at most $c$ transitions, and
  let $R = \min\{R_s \mid s \in G\}$. Then $R > 0$ and if all transition
  probabilities in $\G$ are computed with relative error at most
  $\varepsilon R^3/8(c+1)^2$, then the resulting system
  $(I-A')\vec{V} = \vec{b}'$ has a unique solution $\vec{U}^*$ such that
  $|\vec{V}^*_s - \vec{U}^*_s|/\vec{V}^*_s \leq \varepsilon$ for every
  $s \in G$.
\end{proposition}

\noindent
Note that the constant $R$ of Proposition~\ref{prop-visiting-approx}
can be bounded from below by $x_t^{|Q|-1} \cdot x_n$, where
\begin{itemize}
\item $x_t = \min\{X_\G(s,s') \mid s,s' \in G\}$, i.e., $x_t$ is the
  minimal probability that is either explicitly used in
  $\A$, or equal to some positive termination probability in $\A$;
\item $x_n = \min\{X_\G(s,s') \mid s \in G, s' \in G_1\}$, i.e., $x_n$ is the
  minimal probability that is either a positive termination probability
  in $\A$, or a positive non-termination probability in the pOC's
  $\A_\cons$ and $\A_\inco$ constructed above.
\end{itemize}

\noindent
Now we need to employ the promised divergence gap theorem, which
bounds a positive non-termination probability in pOC away from zero
(for all $p,q \in Q$, we use $[p,q]$ to denote
the probability of all runs $w$ initiated in $p(1)$ that visit
a configuration $q(k)$, where $k \geq 1$ and the counter stays
positive in all configurations preceding this visit).

\begin{theorem}
\label{thm-gap}
  Let $\A = (Q,\delta^{=0},\delta^{>0},P^{=0},P^{>0})$ be a pOC and
  $\X$ the underlying finite-state Markov chain of $\A$.
  Let $p \in Q$ such that $[p{\uparrow}]>0$. Then there are two
  possibilities:
  \begin{enumerate}
     \item There is $q\in Q$ such that $[p,q]>0$ and $[q{\uparrow}]=1$.
        Hence, $[p{\uparrow}] \geq [p,q]$.
     \item There is a BSCC $\B$ of $\X$ and a state $q$ of $\B$ such
        that $[p,q]>0$, $t > 0$, and $\vec{v}_{q}=\vec{v}_{\max}$
        (here $t$ is the trend, $\vec{v}$ is the vector
        of Proposition~\ref{prop-martingale}, and $\vec{v}_{\max}$
        is the maximal component of~$\vec{v}$; all of these
        are considered in $\B$). Further,
        \[
            [p{\uparrow}]\quad \ge\quad
            [p,q]\cdot \frac{t^3}{12 (2 (\vmax - \vmin) + 4)^3}\,.
        \]
 \end{enumerate}
\end{theorem}
Hence, denoting the relative precision $\varepsilon R^3/8(c+1)^2$
of Proposition~\ref{prop-visiting-approx} by $\delta$, we obtain
that $\log(1/\delta)$ is bounded by a polynomial in $|\A|$ and
$\log(1/\varepsilon)$. Further, the transition probabilities
of $\G$ can be approximated up to the relative error $\delta$
in time polynomial in $|\A|$ and $\log(1/\varepsilon)$ by approximating
the termination probabilities of~$\A$ (see
Proposition~\ref{prop-termination}). This proves
Theorem~\ref{thm-omega-regular}.

%%% Local Variables:
%%% mode: latex
%%% TeX-master: "main"
%%% End:

\section{Experimental results, future work}
\label{sec-concl}

We have implemented a prototype tool in the form of a Maple worksheet%
\footnote{Available at {\scriptsize
    \texttt{http://www.comlab.ox.ac.uk/people/stefan.kiefer/pOC.mws}}.}%
, which allows to compute the termination probabilities of pOC, as
well as the conditional expected termination times.  Our tool employs
Newton's method to approximate the termination probabilities within
a sufficient accuracy so that the expected termination time is computed
with absolute error (at most) one by solving
the linear equation system from Section~\ref{subsec:fintermtime}.

We applied our tool to the pOC from Fig.~\ref{fig-and-or-model} for
various values of the parameters.
Fig.~\ref{fig:numbers} shows the results. We also show the
associated termination probabilities, rounded to three digits.
We write $[a{\downarrow}0]$ etc.\ to abbreviate
$[(\textit{and,init}){\downarrow}(\textit{or,return,0})]$ etc.,
 and $[a{\downarrow}]$ for $[a{\downarrow}0] + [a{\downarrow}1])$.

\begin{figure}[t]
\begin{center}
\begin{tabular}{l@{\quad}|@{\quad}r@{\quad}r@{\quad}r@{\quad}r@{\quad}r}
  & $[a{\downarrow}]$ & $[a{\downarrow}0]$ & $[a{\downarrow}1]$ & $E[a{\downarrow}0]$ & $E[a{\downarrow}1]$ \\
 % & $[o{\downarrow}]$ & $[o{\downarrow}0]$ & $[o{\downarrow}1]$ & $E[o{\downarrow}0]$ & $E[o{\downarrow}1]$ \\
\hline
 $z = 0.5, y = 0.4, x_a = 0.2, x_o = 0.2$ & 0.800 & 0.500 & 0.300 & 11.000 & 7.667 \\
 $z = 0.5, y = 0.4, x_a = 0.2, x_o = 0.4$ & 0.967 & 0.667 & 0.300 & 104.750 & 38.917 \\
 $z = 0.5, y = 0.4, x_a = 0.2, x_o = 0.6$ & 1.000 & 0.720 & 0.280 & 20.368 & 5.489 \\
 $z = 0.5, y = 0.4, x_a = 0.2, x_o = 0.8$ & 1.000 & 0.732 & 0.268 & 10.778 & 2.758 \\
\hline
 $z = 0.5, y = 0.5, x_a = 0.1, x_o = 0.1$ & 0.861 & 0.556 & 0.306 & 11.400 & 5.509 \\
 $z = 0.5, y = 0.5, x_a = 0.2, x_o = 0.1$ & 0.931 & 0.556 & 0.375 & 23.133 & 20.644 \\
 $z = 0.5, y = 0.5, x_a = 0.3, x_o = 0.1$ & 1.000 & 0.546 & 0.454 & 83.199 & 111.801 \\
 $z = 0.5, y = 0.5, x_a = 0.4, x_o = 0.1$ & 1.000 & 0.507 & 0.493 & 12.959 & 21.555 \\
\hline
 $z = 0.2, y = 0.4, x_a = 0.2, x_o = 0.2$ & 0.810 & 0.696 & 0.115 & 7.827 & 6.266 \\
 $z = 0.3, y = 0.4, x_a = 0.2, x_o = 0.2$ & 0.811 & 0.636 & 0.175 & 8.928 & 6.783 \\
 $z = 0.4, y = 0.4, x_a = 0.2, x_o = 0.2$ & 0.808 & 0.571 & 0.236 & 10.005 & 7.258 \\
 $z = 0.5, y = 0.4, x_a = 0.2, x_o = 0.2$ & 0.800 & 0.500 & 0.300 & 11.000 & 7.667 \\
\end{tabular}
\end{center}
\caption{Quantities of the pOC from Fig.~\ref{fig-and-or-model}}
\label{fig:numbers}
\end{figure}

We believe that other interesting quantities and numerical
characteristics of pOC,
related to both finite paths and infinite runs, can also be
efficiently approximated using the methods developed in this paper.
An efficient
implementation of the associated algorithms would result in a
verification tool capable of analyzing an interesting class
of infinite-state stochastic programs, which is beyond the
scope of currently available tools limited to finite-state
systems only.

%%% Local Variables:
%%% mode: latex
%%% TeX-master: "main"
%%% End:

%\input{ltl-comments}
%\input{gap-theorem}
%\input{conclusions}

\bibliographystyle{plain}
\bibliography{db,str-short,concur}
%\end{document}

\newpage
\appendix
\section{Proofs}
In this section we give the proofs that were omitted in the main
body of the paper. The appendix is structured according to sections
and subsections of the main part.

\subsection{Finiteness of the expected termination
time (Section~\ref{subsec:inftermtime})} \label{app:subsec-inftermtime}

Recall that $\A = (Q,\delta^{=0},\delta^{>0},P^{=0},P^{>0})$ is
a fixed pOC, $\X$ is the underlaying Markov chain of $\A$, and
$A$ is the transition matrix of~$\X$.

%\smallskip
%\noindent

This section has two parts.
In the first part (Section~\ref{app:subsec-inftermtime}.1) we provide the proofs that apply specifically
 to the case where $\X$ is strongly connected.
In the second part (Section~\ref{app:subsec-inftermtime}.2) we deal with the general case,
 showing Theorem~\ref{thm-exp-infinite}.

\subsubsection{\ref{app:subsec-inftermtime}.1 Strongly connected $\X$}
\ \\[3mm]

\noindent
Recall that
\begin{itemize}
\item $\valpha \in (0,1]^Q$ is the invariant distribution of~$\X$,
\item $\vs \in \Rset^Q$ is the vector expected counter changes defined
  by
  \[
    \vs_p = \sum_{(p,c,q) \in \delta^{>0}} P^{>0}(p,c,q) \cdot c
  \]
\item $t$ is the trend of $\X$ given by $t = \valpha \vs$.
\end{itemize}

\noindent
A \emph{potential} is any vector $\vv$ that satisfies
$\vs + A \vv = \vv + \vone t$. The intuitive meaning of
a potential~$\vv$ is that, starting in any state $p \in Q$,
the expected counter increase after $i$ steps for large~$i$ is $i t + \vv_p$.
%In Lemma~\ref{lem:martingale} below, $\vv$ is used to construct a martingale.
Given a potential~$\vv$, we define $\va := \vmax - \vmin$,
where $\vmax$ and $\vmin$ are the largest and the smallest
component of~$\vv$, respectively. Now we prove two lemmata that together
imply Proposition~\ref{prop-martingale}.
\begin{lemma} \label{lem:v}
 We have the following:
 \begin{itemize}
  \item[(a)]
   Let $W := \vone \valpha$, i.e., each row of~$W$ equals~$\valpha$.
   Let $Z := (I - A + W)^{-1}$.
   The matrix $Z$ exists and the vector $Z \vs$ is a potential.
  \item[(b)]
   Denote by~$\xmin$ the smallest nonzero coefficient of~$A$.
   There exists a potential~$\vv$ with $\va \le 2 |Q| / \xmin^{|Q|}$.
 \end{itemize}
\end{lemma}
\begin{proof}~
 \begin{itemize}
  \item[(a)]
   The matrix $Z := (I - A + W)^{-1}$ exists by \cite[Theorem 5.1.3]{KS60}.
   (The matrix~$Z$ is sometimes called the \emph{fundamental matrix} of the finite Markov chain induced by~$A$.)
   Furthermore, by \cite[Theorem 5.1.3(d)]{KS60} the fundamental matrix~$Z$ satisfies $I + A Z = Z + W$.
   Multiplying with~$\vs$ and setting $\vu := Z \vs$, we obtain
    $\vs + A \vu = \vu + \vone \valpha \vs$; i.e., $Z \vs$ is a potential.
  \item[(b)]
   Let $\vu$ be the potential from~(a); i.e., we have
   \begin{equation}
     (I - A) \vu = \vs - \vone t \,. \label{eq:lem-bound-v-lin-eq}
   \end{equation}
   By the Perron-Frobenius theorem for strongly connected matrices, there exists a positive vector $\vd \in (0,1]^Q$ with $A \vd = \vd$;
    i.e., $(I - A) \vd = \vzero$.
   Observe that $\vu + r \vd$ is a potential for all $r \in \Rset$.
   Choose $r$ such that $\vv := \vu + r \vd$ satisfies $\vmax = 2 |Q| / \xmin^{|Q|}$.
   It suffices to prove $\vmin \ge 0$.
   Let $q \in Q$ such that $\vv_q = \vmax$.
   Define the \emph{distance} of a state $p \in Q$ as the distance of~$p$ from~$q$ in the graph induced by~$A$.
   Note that $q$ has distance~$0$ and all states have distance at most $n-1$, as $A$ is strongly connected.
   We prove by induction that a state~$p$ with distance~$i$ satisfies $\vv_p \ge 2 (n-i) / \xmin^{n-i}$.
   The claim is obvious for the induction base ($i=0$).
   For the induction step, let $p$ be a state with distance~$i+1$ and $i \ge 0$.
   Let $r$ be a state with distance $i$ and $A_{p r} > 0$.
   We have:
   \begin{align*}
    \vv_p & =   (A \vv)_p + \vs_p - t && \text{(as $\vv$ is a potential)} \\
          & \ge (A \vv)_p - 2         && \text{(as $\vs_p, t \in [-1,1]$)} \\
          & \ge \xmin \vv_r - 2       && \text{(as $A_{p r} > 0$ implies $A_{p r} \ge \xmin$)} \\
          & \ge \xmin \cdot 2 (n-i) / \xmin^{n-i} - 2 && \text{(by induction hypothesis)} \\
          & =   2 (n-i) / \xmin^{n-(i+1)} - 2 \\
          & \ge 2 (n-(i+1)) / \xmin^{n-(i+1)}         && \text{(as $\xmin \le 1$)} \,.
   \end{align*}
   This completes the induction step.
   Hence we have $\vmin \ge 0$ as desired.
 \end{itemize}
\qed
\end{proof}

\noindent In the following, the vector~$\vv$ is always a potential.
Recall that $\ps{i}$ and $\cs{i}$ are random variables
which to every run $w \in \run(r(c))$ assign the control state
and the counter value of the configuration $w(i)$, respectively, and
$\ms{i}$ is a random variable defined by
\[
 \ms{i}    = \begin{cases}
              \cs{i} + \vv_{\ps{i}} - i t & \text{if $\cs{j} \ge 1$ for all $0 \le j < i$} \\
              \ms{i-1} & \text{otherwise}
             \end{cases}
\]

\begin{lemma} \label{lem:martingale}
% Let $A$ be strongly connected. % and $t > -\infty$.
 The sequence $\ms{0}, \ms{1}, \ldots$ is a martingale.
\end{lemma}
\begin{proof}
 Fix a path $u \in \fpath(\ps{0}(\cs{0}))$ of length $i \ge 1$. % and let $w$ be an arbitrary run of~$\run(u)$.
 First assume that $\cs{j} \ge 1$ does not hold for all $j \in \{0,\ldots,i-1\}$.
 Then for every run $w \in \run(u)$ we have $\ms{i}(w) = \ms{i-1}(w)$.
 Now assume that $\cs{j} \ge 1$ holds for all $j \in \{0,\ldots,i-1\}$.
 Then we have:
 \begin{align*}
  \E\left[ \ms{i} \;\middle\vert\; \run(u) \right]
  & = \E\left[ \cs{i} + \vv_{\ps{i}} - i t \;\middle\vert\; \run(u) \right] \\
  & = \cs{i-1} + \mathop{\sum_{(\ps{i-1},a,q) \in \delta^{>0}}}_{P^{>0}(\ps{i-1},a,q)=x} x \cdot a
               + \mathop{\sum_{(\ps{i-1},a,q) \in \delta^{>0}}}_{P^{>0}(\ps{i-1},a,q)=x} x \cdot \vv_q - i t \\
  & = \cs{i-1} + \vs_{\ps{i-1}} + \left( A \vv \right)_{\ps{i-1}} - i t \\
  & = \ms{i-1} + \vs_{\ps{i-1}} + \left( A \vv \right)_{\ps{i-1}} - \vv_{\ps{i-1}} - t \\
  & = \ms{i-1} \,,
 \end{align*}
 where the last equality holds because $\vv$ is a potential.
\qed
\end{proof}

\noindent
A direct corollary to Lemma~\ref{lem:v} and Lemma~\ref{lem:martingale}
is the following:

\begin{refproposition}{prop-martingale}
  There is a vector $\vv \in \Rset^{Q}$ such that the stochastic process
  $m^{(1)},m^{(2)},\dots$ defined by
  \[
    \ms{i} = \begin{cases}
            \cs{i} \ +\ \vv_{\ps{i}}  \ -\  i \cdot t &
                 \text{if $\cs{j} \ge 1$ for all $0 \leq j < i$;}\\
            \ms{i-1} & \text{otherwise}
          \end{cases}
  \]
  is a martingale, where $t$ is the trend of $\X$.

  Moreover, the vector $\vv$ satisfies
  $\vmax-\vmin\  \le\  2 |Q| / \xmin^{|Q|}$, where $\xmin$ is the smallest
  positive transition probability in~$\X$, and $\vmax$ and $\vmin$
  are the maximal and the minimal components of $\vv$, respectively.
\end{refproposition}

\noindent
Now we prove the propositions needed to justify Claims~(A) and~(B)
of Section~\ref{subsec:inftermtime}.

\begin{refproposition}{lem:expected-term-bound-prob}
 Let $p(k)$ be an initial configuration, and let $H_i$ be set of all
 runs initiated in $p(k)$ that visit a configuration with zero counter
 in exactly $i$~transitions.
 % i.e., $H_i := \{w \in \run(p(k)) \mid \cs{i} = 0 \ \land \ \forall 0 \le j < i: \cs{j} \ge 1\}$.
 Let
  \[
   a = \exp\left(- \frac{t^2}{8 (\va + t + 1)^2} \right)\,.
  \]
 Note that $0 < a < 1$. Further, let
  \[
   h = \begin{cases} 2 \cdot \frac{-\va - \cs{0}}{t} & \text{if \ $t < 0$} \\
                      2 \cdot \frac{ \va - \cs{0}}{t} & \text{if \ $t > 0$ .}
        \end{cases}
  \]
 Then for all $i \in \Nset$ with $i \ge h$ we have that $\calP(H_i) \le a^i$.
\end{refproposition}
\begin{proof}
 For all runs in~$H_i$ we have $\ms{i} = \vv_{\ps{i}} - i t$ and so
  \begin{equation}
   \ms{0} - \ms{i} = \cs{0} + \vv_{\ps{0}} - \vv_{\ps{i}} + i t \,. \label{eq:lem-exp-term-bound-mart-hit}
  \end{equation}
 \begin{description}
  \item[Case $t < 0$:]
   By~\eqref{eq:lem-exp-term-bound-mart-hit} we have for $i \ge h$:
   \begin{align*}
    \calP(H_i)
    & = \calP(H_i \ \land \ \ms{i} - \ms{0} = -\cs{0} - \vv_{\ps{0}} + \vv_{\ps{i}} - i t) \\
    & \le \calP(\ms{i} - \ms{0} = -\cs{0} - \vv_{\ps{0}} + \vv_{\ps{i}} - i t) \\
    & \le \calP(\ms{i} - \ms{0} \ge -\cs{0} - \va - i t)  \\
    &  =  \calP\left(\ms{i} - \ms{0} \ge (i - h/2) \cdot (-t)\right) \\
    & \le \calP\left(\ms{i} - \ms{0} \ge (i/2) \cdot (-t)\right) \,.
   \end{align*}
  \item[Case $t > 0$:]
   By~\eqref{eq:lem-exp-term-bound-mart-hit} we have for $i \ge h$:
   \begin{align*}
    \calP(H_i)
    & = \calP(H_i \ \land \ \ms{0} - \ms{i} = \cs{0} + \vv_{\ps{0}} - \vv_{\ps{i}} + i t) \\
    & \le \calP(\ms{0} - \ms{i} = \cs{0} + \vv_{\ps{0}} - \vv_{\ps{i}} + i t) \\
    & \le \calP(\ms{0} - \ms{i} \ge \cs{0} - \va + i t)  \\
    &  =  \calP\left(\ms{0} - \ms{i} \ge (i - h/2) \cdot t \right) \\
    & \le \calP\left(\ms{0} - \ms{i} \ge (i/2) \cdot t \right) \,.
   \end{align*}
 \end{description}
 In each step, the martingale value changes by at most $\va + t + 1$.
 Hence Azuma's inequality (see~\cite{Williams:book}) asserts for $t \ne 0$ and $i \ge h$:
 \begin{align*}
  \calP(H_i)
  & \le \exp  \left(- \frac{(i/2)^2 t^2}{2 i (\va + t + 1)^2} \right) && \text{(Azuma's inequality)} \\
  & = a^i \,.
 \end{align*}
\qed
\end{proof}

\begin{refproposition}{prop-term-inf}
  Assume that $\pre^*(q(0))$ is infinite. Then almost all runs initiated
  in an arbitrary configuration reach $Q(0)$. Moreover, there is
  $k_1\in \Nset$ such that, for all $\ell\geq k_1$, the expected
  length of an honest path from $r(\ell)$ to $Q(0)$ is infinite.
%
%  If $\pre^*(q(0))$, then there is a $k_1 \in \Nset$ such that for
%  every configuration $r(\ell)$ where $\ell \geq k_1$ we have that
%  $\calP(\run(r(\ell){\downarrow})) = 1$ and
%  $E(r(\ell){\downarrow}) = \infty$.
\end{refproposition}
\begin{proof}
%
%{\bf
%The following proof shows that for $\ell\geq \lceil\vmax\rceil$, the expected time to reach $0$ is infinite. I would suggest such formulation for the proposition but then other proofs must be adapted. THE PROOF WAS COMPILED IN HASTE! PLEASE CHECK CORRECTNESS!!
%}
As $\pre^*(q(0))=\infty$ and $\X$ is strongly connected, $Q(0)$ is reachable from every configuration with positive probability.
\tomas{NOT SURE, PLEASE CHECK!!}
%\stefan{It depends on $t=0$, right? If yes, then say it. (I don't want to dig in your SODA-techrep.)}
Also, recall that $t=0$.
Using strong law of large numbers (see~e.g.~\cite{Williams:book}) and results of~\cite{BBEKW:OC-MDP-arXiv} (in particular Lemma~19),
one can show that $Q(0)$ is reached from any configuration with probability one.

Consider an initial configuration $r(\ell)$ with $\ell + \vv_r > \vmax$.
We will show that the expected length of an honest path from~$r(\ell)$ to~$Q(0)$ is infinite; i.e., we can take $k_1 := \lceil \va + 1\rceil$.
Consider the martingale $m^{(1)},m^{(2)},\dots$ defined in Proposition~\ref{prop-martingale} over $\run(r(\ell))$.
Note that as $t=0$, the definition of the martingale simplifies to
\[
    \ms{i} = \begin{cases}
            \cs{i} \ +\ \vv_{\ps{i}} &
                 \text{if $\cs{j} \ge 1$ for all $0 \leq j < i$;}\\
            \ms{i-1} & \text{otherwise}
          \end{cases}
\]
% Denote by $M$ the number $\vmax+1$ (note that the martingale value changes by at most $M$) and
Observe that $\ms{0} = \ell + \vv_{r}$ and that the martingale value changes by at most~$M := \lceil\va\rceil + 1$ in a single step.
Let us fix $k\in \Nset$ such that $\ell + \vv_r < \vmax + k$.
Define a {\em stopping time}~$\tau$ (see e.g.~\cite{Williams:book})
  which returns the first point in time in which either $m^{(\tau)}\geq \vmax + k$, or $m^{(\tau)}\leq \vmax$.
Observe that $\tau$ is almost surely finite and that $m^{(\tau)}\in
[\vmax - M, \vmax]\cup [\vmax + k, \vmax + k + M]$.
Define $x := \calP(\ms\tau \ge \vmax + k)$.
Then
\begin{eqnarray}\label{eq:expboundmart}
\E[\ms\tau] &\ \leq\ & x \cdot (\vmax+k+M)  + (1-x) \cdot \vmax = \vmax + x \cdot (k+M)
\end{eqnarray}
and by the optional stopping theorem (see e.g.~\cite{Williams:book}),
\begin{equation}\label{eq:expstart}
\E[\ms\tau]\ =\ \E[\ms0]\ =\ \ell+\vv_r \,.
\end{equation}
By putting the equations~(\ref{eq:expboundmart})~and~(\ref{eq:expstart}) together, we obtain that
\begin{equation}
\calP(\ms\tau \geq \vmax + k)\quad\geq\quad \frac{\ell + \vv_r - \vmax}{k+M} \,. \label{eq:mart-opt-stop-upper}
\end{equation}
Denote by~$T$ the time to hit~$Q(0)$.
We need to show $\E T = \infty$.
For any run~$w$ with $\ms\tau \ge \vmax + k$ we have
 \[
  \cs\tau = \ms\tau - \vv_{\ps\tau} \ge \vmax + k - \vv_{\ps\tau} \ge k \;,
 \]
 hence we have $T \ge k$ for~$w$, as at least $k$ steps are required to decrease the counter value from~$k$ to~$0$.
It follows $\calP(\ms\tau \ge \vmax + k) \le \calP(T \ge k)$.
Hence:
\begin{align*}
 \E T & = \sum_{k \in \Nset} \calP(T \ge k) \ge \sum_{k=\ell + 1}^\infty \calP(T \ge k)  \\
      & \ge \sum_{k=\ell + 1}^\infty \calP(\ms\tau \ge \vmax + k)
       \mathop{\ge}^{\eqref{eq:mart-opt-stop-upper}} \sum_{k=\ell + 1}^\infty \frac{\ell + \vv_r - \vmax}{k+M} = \infty \,.
\end{align*}
\qed
\end{proof}

\begin{refproposition}{prop-pre-closed}
  There is $k_2 \in \Nset$ such that for every configuration
  $r(\ell) \in \pre^*(q(0))$, where $\ell \geq k_2$, we have that
  if $r(\ell) \tran{} r'(\ell')$, then $r'(\ell') \in \pre^*(q(0))$.
\end{refproposition}
\begin{proof}
  We start by observing that $\pre^*(q(0))$ has an ``ultimately periodic''
  structure. For every $i \in \Nseto$, let 
  $\pre(i) = \{ r \in Q \mid r(i) \in  \pre^*(q(0))\}$. Note
  that if $\pre(i) = \pre(j)$ for some $i,j \in \Nseto$, then also
  $\pre(i{+}1) = \pre(j{+}1)$. Let $m_1$ be the least index such that
  $\pre(m_1) = \pre(j)$ for some $j>m_1$, and let $m_2$ be the least
  $j$ with this property. Further, we put $m = m_2 - m_1$. Observe that
  $m_1,m_2 \leq 2^{|Q|}$, and for every $\ell \geq m_2$ we have that
  $\pre(\ell) = \pre(\ell{+}m)$. 

  For every configuration $r(\ell)$ of $\A$, let $C(r(\ell))$ be the
  set of all configurations $r(\ell+i)$ such that $0\leq i < m$ and
  $r \in \pre(\ell{+}i)$. Note that $C(r(\ell))$ has at most $m$ elements,
  and we define the \emph{index} of $r(\ell)$ as the cardinality
  of $C(r(\ell))$. Due the periodicity of $\pre^*(q(0))$,
  we immediately obtain that for every $r(\ell)$ and $j \in \Nseto$, 
  where $\ell \geq m_1$, the index of $r(\ell)$ is the same as the
  index of  $r(\ell{+}j)$.

  Let $k_2 = m_1 + |Q| + 1$, and assume that there is a transition 
  $r(\ell) \tran{} r'(\ell')$ such that $r \in \pre(\ell)$, 
  $r' \not\in \pre(\ell')$, and $\ell \geq k_2$. Then 
  $r(\ell{+}i) \tran{} r'(\ell'{+}i)$ for all $0 \leq i  <m$. Obviously,
  if $r' \in \pre(\ell'{+}i)$, then also $r \in \pre(\ell{+}i)$, which
  means that the index of $r'(\ell')$ is \emph{strictly smaller}
  that the index of $r(\ell)$. Since $\X$ is strongly 
  connected, there is finite path from $r'(\ell')$ to $r(n)$ of length
  at most $|Q|$, where $n \geq m_1$. This means that there is 
  a finite path from $r'(\ell'{+}i)$ to $r(n{+}i)$ for every $0 \leq i <m$.
  Hence, the index of $r'(\ell')$ is at least as large as the index
  of $r(n)$. Since the indexes of $r(n)$ and $r(\ell)$ are the same,
  we have a contradiction.
\qed
\end{proof}

%%% Local Variables:
%%% mode: latex
%%% TeX-master: "main"
%%% End:

\subsubsection{\ref{app:subsec-inftermtime}.2 General Case}
\ \\[3mm]

\begin{lemma} \label{lem:hitting-time-finite-chain}
 Consider a finite Markov chain on a set~$Q$ of states with $|Q| = n$. %and transition matrix $A \in [0,1]^{Q \times Q}$.
 Let $x$ denote the smallest nonzero transition probability in the chain.
 Let $p \in Q$ be any state and $S \subseteq Q$ any subset of~$Q$.
 Define the random variable $T$ on runs starting in~$p$ by
  \[
   T := \begin{cases} k & \text{if the run hits a state in~$S$ for the first time after exactly $k$ steps} \\
                      \mathit{undefined} & \text{if the run never hits a state in~$S$ .}
        \end{cases}
  \]
% If $x = 1$, then $\calP(T \ge n) = 0$.
 We have $\calP(T \ge k) \le 2 c^k$ for all $k \ge n$, where $c := \exp(-x^n/n)$.
\end{lemma}
\begin{proof}
 If $x=1$ then all states that are visited are visited after at most $n-1$ steps and hence $\calP(T \ge n) = 0$.
 Assume $x < 1$ in the following.
 Since for each state the sum of the probabilities of the outgoing edges is~$1$, we must have $x \le 1/2$.
 Call \emph{crash} the event of, within the first $n-1$ steps, either hitting~$S$ or some state~$r \in Q$ from which $S$ is not reachable.
 The probability of a crash is at least $x^{n-1} \ge x^n$, regardless of the starting state.
 Let $k \ge n$.
 For the event where $T \ge k$, a crash has to be avoided at least $\lfloor \frac{k-1}{n-1} \rfloor$ times; i.e.,
 \[
  \calP(T \ge k) \le (1-x^n)^{\lfloor \frac{k-1}{n-1} \rfloor} \,.
 \]
 As $\lfloor \frac{k-1}{n-1} \rfloor \ge \frac{k-1}{n-1} - 1 \ge \frac{k}{n} - 1$, we have
 \begin{align*}
  \calP(T \ge k)
  & \le \frac{1}{1-x^n} \cdot \left((1-x^n)^{1/n}\right)^k
    \le 2 \cdot \left((1-x^n)^{1/n}\right)^k \\
  &  =  2 \cdot \exp\left(\frac1n \log(1-x^n)\right)^k
    \le 2 \cdot \exp\left(\frac1n \cdot (-x^n)\right)^k
    = 2 \cdot c^k \,.
 \end{align*}
\qed
\end{proof}

\newcommand{\Rs}[1]{R^{(#1)}}%
\begin{lemma} \label{lem:etime-case-C}
 Let $p, q \in Q$ such that $[p{\downarrow}q] > 0$ and $q$ is not in a BSCC of~$\X$.
 Then
 \[
  E(p{\downarrow}q) \le \frac{5 |Q|}{\xmin^{|Q| + |Q|^3}} \,.
 \]
\end{lemma}
\begin{proof}
 Consider the finite Markov chain~$\X$.
 Define, for runs in~$\X$ starting in~$p$, the random variable $\widehat{R}$ as the time to hit~$q$,
  and set $\widehat{R} := \mathit{undefined}$ for runs that do not hit~$q$.
 There is a straightforward probability-preserving mapping that maps runs in~$\M_\A$ with $R_{p{\downarrow}q} = k$ to runs in~$\X$ with $\widehat{R} = k$.
 Hence, $\calP(R_{p{\downarrow}q} = k) \le \calP(\widehat{R} = k)$ for all $k \in \Nset_0$ and so
 \begin{align*}
  E(p{\downarrow}q) \cdot [p{\downarrow}q]
  &  =  \sum_{k \in \Nset_0} \calP(R_{p{\downarrow}q} = k) \cdot k
    \le \sum_{k \in \Nset_0} \calP(\widehat{R} = k) \cdot k \\
  &  =  \sum_{k \in \Nset} \calP(\widehat{R} \ge k)
    \le \sum_{k=1}^{|Q|} 1 + \sum_{k=0}^\infty 2 c^k  = |Q| + \frac{2}{1 - c} && \text{(Lemma~\ref{lem:hitting-time-finite-chain})\,.}
\intertext{We have $1-c = 1 - \exp(-\xmin^{|Q|}/|Q|) \ge \xmin^{|Q|}/(2 |Q|)$, hence}
  E(p{\downarrow}q) \cdot [p{\downarrow}q]
  & \le |Q| + \frac{4 |Q|}{\xmin^{|Q|}} \le \frac{5 |Q|}{\xmin^{|Q|}} \,.
 \end{align*}
 As $[p{\downarrow}q] \ge \xmin^{|Q|^3}$ by Proposition~\ref{prop:termprobs}, it follows
 \[
  E(p{\downarrow}q) \le \frac{5 |Q|}{\xmin^{|Q| + |Q|^3}} \,.
 \]
\qed
\end{proof}

\begin{lemma} \label{lem:etime-case-D}
 Let $p, q \in Q$ such that $[p{\downarrow}q] > 0$ and $q$ is in a BSCC with trend $t \ne 0$.
 Then
 \[
  E(p{\downarrow}q) \le 85000 \cdot \frac{|Q|^6}{\xmin^{5 |Q| + |Q|^3} \cdot t^4} \,.
 \]
\end{lemma}
\begin{proof}
 Let $B$ denote the BSCC of~$q$.
 For a run $w \in \run(p{\downarrow}q)$, define $\Rs1(w)$ as the time to hit~$B$, and $\Rs2(w)$ as the time to reach $q(0)$ after hitting~$B$.
 For other runs~$w$ let $\Rs1(w) := \mathit{undefined}$ and $\Rs2(w) := \mathit{undefined}$.
 Note that $R_{p{\downarrow}q}(w) = \Rs1(w) + \Rs2(w)$ whenever $\Rs1(w)$ and $\Rs2(w)$ are defined.
 We have:
 \begin{align*}
  E(p{\downarrow}q) \cdot [p{\downarrow}q]
  & = \sum_{k \in \Nset_0} \calP(R_{p{\downarrow}q} = k) \cdot k \\
  & = \sum_{k \in \Nset_0} \calP(\Rs1 + \Rs2 = k) \cdot k \\
  & = \sum_{k_1, k_2 \in \Nset_0} \calP(\Rs1 = k_1 \ \land \ \Rs2 = k_2) \cdot (k_1 + k_2) \\
  & = \sum_{k_1, k_2 \in \Nset_0} \calP(\Rs1 = k_1) \cdot \calP(\Rs2 = k_2 \mid \Rs1 = k_1) \cdot (k_1 + k_2) \\
  & = E_1 + E_2 \,,
 \end{align*}
 where
 \begin{align*}
   E_1 & := \sum_{k_1, k_2 \in \Nset_0} \calP(\Rs1 = k_1) \cdot \calP(\Rs2 = k_2 \mid \Rs1 = k_1) \cdot k_1 \qquad \text{and} \\
   E_2 & := \sum_{k_1, k_2 \in \Nset_0} \calP(\Rs1 = k_1) \cdot \calP(\Rs2 = k_2 \mid \Rs1 = k_1) \cdot k_2 \,.
 \end{align*}
 For a bound on~$E_1$ we have
 \begin{align*}
   E_1 & = \sum_{k_1 \in \Nset_0} \calP(\Rs1 = k_1) \cdot k_1 \cdot \sum_{k_2 \in \Nset_0} \calP(\Rs2 = k_2 \mid \Rs1 = k_1) \\
       & \le \sum_{k_1 \in \Nset_0} \calP(\Rs1 = k_1) \cdot k_1
 \end{align*}
 Consider the finite Markov chain~$\X$.
 Define, for runs in~$\X$ starting in~$p$, the random variable $\widehat{\Rs1}$ as the time to hit~$B$,
  and set $\widehat{\Rs1} := \mathit{undefined}$ for runs that do not hit~$B$.
 There is a straightforward probability-preserving mapping that maps runs in~$\M_\A$ with $\Rs1 = k_1$ to runs in~$\X$ with $\widehat{\Rs1} = k_1$.
 Hence, $\calP(\Rs1 = k_1) \le \calP(\widehat{\Rs1} = k_1)$ for all $k_1 \in \Nset_0$ and so
 \begin{equation} \label{eq:bound-E1}
  E_1  \le \sum_{k_1 \in \Nset_0} \calP(\widehat{\Rs1} = k_1) \cdot k_1
          \quad = \quad \sum_{k_1 \in \Nset} \calP(\widehat{\Rs1} \ge k_1)
          \quad \le \quad \frac{2}{1-c}
 \end{equation}
 with $c$ from Lemma~\ref{lem:hitting-time-finite-chain}.

 For a bound on $E_2$, fix any $k_1 \in \Nset_0$.
 We have:
 \begin{align*}
  & \quad \sum_{k_2 \in \Nset_0} \calP(\Rs2 = k_2 \mid \Rs1 = k_1) \cdot k_2 \\
%  & = \sum_{k_2 \in \Nset_0} \calP(\Rs2 \ge k_2 \mid \Rs1 = k_1) \\
  & = \sum_{j=0}^{k_1+1} \sum_{k_2 \in \Nset_0} \underbrace{\calP(\Rs2 = k_2 \mid \Rs1 = k_1, \ \cs{0} = j)}_{= \calP(\Rs2 = k_2 \mid \cs{0} = j)}  \cdot  k_2
              \cdot \calP(\cs{0} = j \mid \Rs1 = k_1) \,,
\intertext{%
 where we denote by~$\cs{0}$ the counter value when hitting~$B$.
 In the last equality we used the fact that in each step the counter value can increase by at most~$1$, thus $\Rs1 = k_1$ implies $\cs{0} \le k_1+1$.
 Denote by $m(k_1) \in \{0, \ldots, k_1+1\}$ the value of~$j$ that maximizes
  $\sum_{k_2 \in \Nset_0} \calP(\Rs2 = k_2 \mid \cs{0} = j) \cdot  k_2$.
 Then we can continue:
}
 & \le \sum_{k_2 \in \Nset_0} \calP(\Rs2 = k_2 \mid \cs{0} = m(k_1)) \cdot k_2 \cdot \underbrace{\sum_{j=0}^{k_1+1} \calP(\cs{0} = j \mid \Rs1 = k_1)}_{=1}
\intertext{%
 Denote by $h(\cs{0})$ the $h$ from Lemma~\ref{lem:expected-term-bound-prob}.
 We have $h(m(k_1)) \le 2 \frac{\va + m(k_1)}{|t|} \le 2 \frac{\va + k_1 + 1}{|t|} =: \hat h(k_1)$.
 So we can continue:
}
 & \le \sum_{k_2 = 0}^{\left\lfloor \hat h(k_1)\right\rfloor} k_2 + \sum_{k_2 = \left\lceil\hat h(k_1)\right\rceil}^\infty a^{k_2} \cdot k_2
       \qquad \text{(with $a$ from Proposition~\ref{lem:expected-term-bound-prob})} \\
 & \le \hat h(k_1)^2 + \frac{a}{(1-a)^2}
     = \frac{4 (\va + k_1 + 1)^2}{t^2} + \frac{a}{(1-a)^2} \,.
 \end{align*}
 With this inequality and the random variable $\widehat{\Rs2}$ from above at hand we get a bound on~$E_2$:
  \begin{align*}
   E_2 & = \sum_{k_1\in \Nset_0} \underbrace{\underbrace{\calP(\Rs1 = k_1)}_{\Large \le \calP(\widehat{\Rs1}=k_1)}}_{\le \calP(\widehat{\Rs1}\ge k_1)}
    \cdot {\underbrace{\sum_{k_2 \in \Nset_0} \calP(\Rs2 = k_2 \mid \Rs1 = k_1) \cdot k_2}_{{\Large \le \frac{4 (\va + k_1 + 1)^2}{t^2} + \frac{a}{(1-a)^2}}}}  \\
       & \le \sum_{k_1=0}^{|Q|-1}\left( \frac{4 (\va + k_1 + 1)^2}{t^2} + \frac{a}{(1-a)^2}\right)
           + \sum_{k_1=0}^\infty 2 c^{k_1} \frac{a}{(1-a)^2}
           + \sum_{k_1=0}^\infty 2 c^{k_1} \frac{4 (\va + k_1 + 1)^2}{t^2} \\
       & \le \frac{4 |Q| (\va + |Q|)^2}{t^2} + \frac{2 |Q|}{(1-c)(1-a)^2}
           + \frac{8}{t^2} \sum_{k_1=0}^\infty c^{k_1} (\va + k_1 + 1)^2
 \end{align*}
 The last series can be bounded as follows:
 \begin{align*}
  \sum_{k_1=0}^\infty c^{k_1} (\va + k_1 + 1)^2
  & \le \sum_{k_1=0}^{\left\lfloor \va + 1 \right\rfloor} \left(2 (\va + 1)\right)^2
      + \sum_{k_1=\left\lfloor \va + 1 \right\rfloor + 1}^\infty c^{k_1} \cdot (2 k_1)^2 \\
  & \le 4 (\va + 2)^3 + 4 \sum_{k_1=0}^\infty c^{k_1} \cdot k_1^2
     =  4 (\va + 2)^3 + 4 \frac{c (c+1)}{(1-c)^3} \\
  & \le 4 (\va + 2)^3 + \frac{8}{(1-c)^3}
 \end{align*}
 It follows:
 \begin{equation} \label{eq:bound-E2}
   E_2 \le \frac{4 |Q| (\va + |Q|)^2}{t^2} + \frac{2 |Q|}{(1-c)(1-a)^2}
    + \frac{32}{t^2} \left( (\va + 2)^3 + \frac{2}{(1-c)^3} \right)
 \end{equation}
 Recall the following bounds:
 \begin{align*}
   \va & \le 2 |Q| / \xmin^{|Q|} && \text{(Lemma~\ref{lem:v})} \\
   1-c & = 1 - \exp(-\xmin^{|Q|}/|Q|) \ge \xmin^{|Q|} / (2 |Q|) && \text{(Lemma~\ref{lem:hitting-time-finite-chain})} \\
   1-a & = 1 - \exp\left(- t^2 / \left(8 (\va + 2)^2\right)\right) \ge t^2 / \left(16 (\va + 2)^2\right) && \text{(Proposition~\ref{lem:expected-term-bound-prob})} \\
   [p{\downarrow}q] & \ge \xmin^{|Q|^3} && \text{(Proposition~\ref{prop:termprobs})}
 \end{align*}
 After plugging those bounds into \eqref{eq:bound-E1} and~\eqref{eq:bound-E2} we obtain using straightforward calculations:
 \begin{gather*}
  E_1 \le 4 \frac{|Q|}{\xmin^{|Q|}} \qquad \text{and} \qquad  E_2  \le 84356 \frac{|Q|^6}{\xmin^{5 |Q|} \cdot t^4} \;, \qquad \text{hence} \\
  E(p{\downarrow}q) = \frac{E_1 + E_2}{[p{\downarrow}q]} \le 85000 \cdot \frac{|Q|^6}{\xmin^{5 |Q| + |Q|^3} \cdot t^4} \,.
 \end{gather*}
\qed
\end{proof}

\begin{lemma} \label{lem:pumping-pre-post}
Let $p,q\in Q$.
If $\pre^*(q(0))\cap \post^*(p(1))$ is finite, then
\[
|\pre^*(q(0))\cap \post^*(p(1))|\quad \leq\quad |Q|^2\cdot (|Q|+2)
\]
\end{lemma}
\begin{proof}
In this proof we use some notions and results of~\cite{EHRS:MC-PDA} (in particular, we use the notion of $\calP$-automata as
defined in Section~2.1 of~\cite{EHRS:MC-PDA}).
Consider the pOC as a (non-probabilistic) pushdown system with one letter stack alphabet, say $\Gamma=\{X\}$
(the counter of height $n$ then corresponds to the stack content $X^n$).

A $\calP$-automaton $\A_{q(0)}$ accepting the set of configurations $\{q(0)\}$ can be defined to have the set of states $Q$, no transitions, and $q$ as the only accepting state. Let $\A_{\mathit{pre}^*}$ be the $\calP$-automaton accepting
$\pre^*(q(0))$ constructed using the procedure from Section~4 of \cite{EHRS:MC-PDA}. The automaton $\A_{\mathit{pre}^*}$ has the same set of states, $Q$, as $\A_{q(0)}$.

A $\calP$-automaton $\A_{p(1)}$ accepting the set of configurations $\{p(1)\}$ can be defined to have the set of states $Q\cup \{p_{acc}\}$, one transition $(p,X,p_{acc})$, and $q_{acc}$ as the only accepting state. Let $\A_{\mathit{post}^*}$ be the automaton accepting
$\post^*(p(1))$ constructed using the procedure from Section~6 of \cite{EHRS:MC-PDA}.
The automaton $\A_{\mathit{post}^*}$ has at most $|Q|+2$ states.

Using standard product construction we obtain a $\calP$-automaton $\A$ accepting $\pre^*(q(0))\cap \post^*(p(1))$, which
has $|Q|\cdot (|Q|+2)$ states.
Now note that if $\pre^*(q(0))\cap \post^*(p(1))$ is finite, then a standard pumping argument for finite automata
implies that the length of every word accepted by $\A$ is bounded by $|Q|\cdot (|Q|+2)$.
It follows that there are only $|Q|^2\cdot (|Q|+2)$ configurations in $\pre^*(q(0))\cap \post^*(p(1))$.
\qed
\end{proof}

\begin{lemma} \label{lem:etime-case-E}
 Let $p,q\in Q$ such that $\pre^*(q(0))\cap \post^*(p(1))$ is finite.
 Then
  \[
   E(p{\downarrow}q) \le E(p{\downarrow}q) \le \frac{15 |Q|^3}{\xmin^{4 |Q|^3}}
  \]
\end{lemma}
\begin{proof}
 We construct a finite Markov chain~$\Y$ as follows.
 The states of~$\Y$ are the states in $\pre^*(q(0))\cap \post^*(p(1)) \cup \{o\}$, where $o$ is a fresh symbol.
 In general, the transitions in~$\Y$ are as in the infinite Markov chain~$\M_\A$, with the following exceptions:
  \begin{itemize}
   \item
    all transitions leaving the set~$\pre^*(q(0))\cap \post^*(p(1))$ are redirected to~$o$;
   \item
    all transitions leading to a configuration $r(0)$ with $r \ne q$ are redirected to~$o$;
   \item
    $o$ gets a probability~$1$ self-loop.
  \end{itemize}
 Let $T$ denote the time that a run in~$\Y$ starting from~$p(1)$ hits~$q(0)$ in exactly $k$ steps.
 This construction of~$\Y$ makes sure that $\calP(T=k) = \calP(R_{p{\downarrow}q} = k)$.
 Note that by Lemma~\ref{lem:pumping-pre-post} the chain~$\Y$ has at most $\ell := 3|Q|^3$ states.
 So we have:
 \begin{align*}
   [p{\downarrow}q] \cdot E(p{\downarrow}q)
   & \le \sum_{k \in \Nset} \calP(R_{p{\downarrow}q} \ge k) 
      =  \sum_{k \in \Nset} \calP(T \ge k) \\
   &  =  \sum_{k=1}^{\ell-1} \calP(T \ge k) + \sum_{k=\ell}^\infty \calP(T \ge k) \\
   & \le \ell + \sum_{k=0}^\infty 2 c^k = \ell + \frac{2}{1-c} && \text{(Lemma~\ref{lem:hitting-time-finite-chain})}
\intertext{We have $1-c = 1 - \exp(-\xmin^{\ell}/\ell) \ge \xmin^{\ell}/(2 \ell)$, hence}
   [p{\downarrow}q] \cdot E(p{\downarrow}q)
   & \le 3 |Q|^3 + \frac{12 |Q|^3}{\xmin^{3 |Q|^3}} \le \frac{15 |Q|^3}{\xmin^{3 |Q|^3}} \,,
 \end{align*}
 and so, by Proposition~\ref{prop:termprobs},
 \[
  E(p{\downarrow}q) \le \frac{15 |Q|^3}{\xmin^{4 |Q|^3}} \,.
 \]
\qed
\end{proof}

By combining Lemmata \ref{lem:etime-case-C}, \ref{lem:etime-case-D} and~\ref{lem:etime-case-E}
 we obtain the following proposition, which directly implies Theorem~\ref{thm-exp-infinite}:
\begin{proposition} \label{prop:grand-bound}
 Let $(p,q) \in T^{>0}$.
 Let $\B$ be the SCC of~$q$ in~$\X$.
 Let $\xmin$ denote the smallest nonzero probability in~$A$.
 Then we have:
 \begin{itemize}
  \item If $\pre^*(q(0)) \cap \post^*(p(1))$ is a finite set, then
   $\displaystyle
    E(p{\downarrow}q) \le 15 |Q|^3 / \xmin^{4 |Q|^3}
   $;
  \item otherwise, if $\B$ is not a BSCC of~$\X$, then
   $\displaystyle
    E(p{\downarrow}q) \le 5 |Q| / \left(\xmin^{|Q| + |Q|^3}\right)
   $;
  \item otherwise, if $\B$ has trend $t \ne 0$, then
   $\displaystyle
    E(p{\downarrow}q) \le 85000 |Q|^6 / \left(\xmin^{5 |Q| + |Q|^3} \cdot t^4\right)
   $.
  \item otherwise, $E(p{\downarrow}q)$ is infinite.
 \end{itemize}
\end{proposition}

\subsection{Efficient approximation of finite expected termination time
(Section~\ref{subsec:fintermtime})}

\noindent
We will use the following theorem from numerical analysis
(see, e.g., \cite{EWY:one-counter}):

\begin{theorem}
\label{thm:error}
  Consider a system of linear equations, $B\cdot \vec{V}=\vec{b}$,
  where $B\in \Rset^{n\times n}$ and $\vec{b}\in \Rset^n$.  Suppose
  that $B$ is regular and $\vec{b}\not = \vec{0}$. Let
  $\vec{V}^*=B^{-1}\cdot\vec{b}$ be the unique solution of this system
  and suppose that $\vec{V}^*\not =
  \vec{0}$. %Let $\norm{\cdot}$ be any vector norm and associated matrix norm.
  Denote by $\kappa(B)=\norm{B}\cdot \norm{B^{-1}}$ the condition
  number of $B$. Consider a system of equations $(B+{\Delta})\cdot
  \vec{V}=\vec{b}+\vec{\zeta}$ where ${\Delta}\in \Rset^{n\times n}$
  and $\vec{\zeta}\in \Rset^n$.  If
  $\norm{{\Delta}}<\frac{1}{\norm{B^{-1}}}$, then the system
  $(B+{\Delta})\cdot \vec{V}=\vec{b}+\vec{\zeta}$ has a unique
  solution $\vec{V}^*_{p}$.  Moreover, for every $\delta>0$ satisfying
  $\frac{\norm{\Delta}}{\norm{B}}\leq \delta$ and
  $\frac{\norm{\zeta}}{\norm{b}}\leq \delta$ and $4\cdot \delta\cdot
  \kappa(B)<1$ the solution $\vec{V}^*_{p}$ satisfies
\[
\frac{\norm{\vec{V}^* - \vec{V}^*_{p}}}{\norm{\vec{V}^*}}\quad
\leq\quad 4\cdot \delta\cdot \kappa(B)
\]
\end{theorem}

\begin{proposition}
\label{prop:error}
Consider a system of linear equations, $C\cdot \vec{W}=\vec{c}$, where
$C\in \Rset^{n\times n}$ and $\vec{c}\in \Rset^n$.  Suppose that $C$
is nonsingular and $\vec{c}\not = \vec{0}$.  Let
$\vec{W}^*=C^{-1}\cdot\vec{c}$ be the unique solution of this system.
Let $\norm{\cdot}$ be the $l_\infty$ norm.
% Denote by $\kappa(C)=\norm{C}\cdot \norm{C^{-1}}$ the condition number of $C$.
Consider a system $(C+{\calE}) \cdot \vec{W}=\vec{c}$ where
${\calE}\in \Rset^{n\times n}$.  Let $\norm{C} \le u \ge 1$ and
$\norm{C^{-1}} \le v \ge 1$.  If $\norm{{\calE}}<1/v$, then the system
$(C+{\calE})\cdot \vec{W}=\vec{c}$ has a unique solution
$\vec{W}^*_{p}$.  Moreover, if $\norm{\calE} \le \delta < 1 / (4 u
v)$, then $\vec{W}^*_{p}$ satisfies
  \[
  \frac{\norm{\vec{W}^* - \vec{W}^*_{p}}}{\norm{\vec{W}^*}}\quad
  \leq\quad \delta\cdot 4 u v
  \]
\end{proposition}
\begin{proof}
 We apply Theorem~\ref{thm:error} with
 \[
  B := \left(\begin{matrix} C \ & 0 \\ 0 \ & 1 \end{matrix}\right) \qquad \text{and} \qquad b := \left(\begin{matrix} c \\ 1 \end{matrix} \right)
   \qquad \text{and} \qquad \Delta := \left(\begin{matrix} \calE \ & 0 \\ 0 \ & 0 \end{matrix}\right) \,;
 \]
 i.e., a single equation $x=1$, for a new variable~$x$ is added to the system, without new errors.
 Notice that
 \[
  B^{-1} = \left(\begin{matrix} C^{-1} \ & 0 \\ 0 \ & 1 \end{matrix}\right) \qquad
   \text{and} \qquad \vec{V}^* := \left(\begin{matrix} \vec{W}^* \\ 1 \end{matrix} \right) \,.
 \]
 Further $\norm{B^{-1}} = \max\{1, \norm{C^{-1}}\}$.
 So we have $\norm{\Delta} = \norm{\calE} < 1/v \le 1/\max\{1,\norm{C^{-1}}\} = 1 / \norm{B^{-1}}$.
 Thus, by Theorem~\ref{thm:error} there is a unique solution of $(B+{\Delta})\cdot \vec{V}=\vec{b}$,
  hence $\vec{W}^*_{p}$ is unique too.
 Moreover, we have
 \begin{align*}
  \frac{\norm{\Delta}}{\norm{B}} & = \frac{\norm{\Delta}}{\max\{1,\norm{C}\}} \le \norm{\Delta} = \norm{\calE} \le \delta \qquad \text{and} \\
  4 \cdot \delta \cdot \kappa(B) & = 4 \cdot \delta \cdot \max\{1,\norm{C}\} \cdot \max\{1,\norm{C^{-1}}\} \le 4 \cdot \delta \cdot u \cdot v < 1\,,
 \end{align*}
 so Theorem~\ref{thm:error} implies
 \[
 \frac{\norm{\vec{W}^* - \vec{W}^*_{p}}}{\norm{\vec{W}^*}}\quad \leq\quad 4\cdot \delta\cdot \kappa(B) \quad \le \quad \delta \cdot 4 u v\,.
 \]
\qed
\end{proof}

\noindent With this at hand we can prove Proposition~\ref{prop-exp-approx}:
\begin{refproposition}{prop-exp-approx}
%\label{prop-exp-approx}
  %There is a polynomial $P(x)$ such that the following holds.
  Let $b\in \Rset^+$ satisfy $E(p{\downarrow}q)\leq b$ for all $(p,q) \in \pfterm$.
  For each $\varepsilon$, where $0 < \varepsilon < 1$, let
  $\delta = \varepsilon\, / (12\cdot b^2)$.
  If $\norm{G-H} \le \delta$, then the perturbed system
   $\vec{V} = G \cdot \vec{V} + \vone$
  has a unique solution $\vec{F}$.
  Moreover, we have that
  \[
   |E(p{\downarrow} q) - \vec{F}_{pq}| \quad \leq\quad \varepsilon \qquad \text{for all $(p,q) \in \pfterm$.}
  \]
  Here $\vec{F}_{pq}$ is the component of $\vec{F}$ corresponding to the variable $V(p{\downarrow}q)$.
\end{refproposition}
\begin{proof}
Denote by $\vec{E}$ the vector of expected termination times, i.e., the unique solution of $\calL'$,
 i.e., $\vec{E} = (I-H)^{-1} \vone$.
Recall that %$\vec{E} \geq \vone$ and that
 all components of $\vec{E}$ are finite.

We will apply Proposition~\ref{prop:error} using the following assignments:
 $C=I-H, C+{\calE} =I-G, \vec{c}=\vec{1}, \vec{W}^*=\vec{E}, \vec{W}^*_p=\vec{F}$. %, u=3$ (by~\eqref{eq:norm-I-H}).
To find a suitable~$u$, we need to find a bound on $\norm{I-H}$.
By comparing $\calL'$ with~\eqref{eq:termination-probabilities} it follows that $\norm{H \vone} \le 2$ and hence
 \begin{equation} \label{eq:norm-I-H}
  \norm{I - H} \quad \le \quad 1 + \norm{H} \quad = \quad 1 + \norm{H \vone} \quad \le \quad 3 \ =: \ u \,.
 \end{equation}
Further, we set $v := b$, so we need to show $\norm{(I-H)^{-1}} \le b$.
By our assumption, $\norm{\vec{E}} \ \leq\  b$.
%So it suffices to show $\norm{\vec{E}} = \norm{(I-H)^{-1}}$.
Recall that $\vec{E} = (I-H)^{-1} \vone$, so if $(I-H)^{-1}$ is nonnegative, then
 $\norm{(I-H)^{-1}} = \norm{(I-H)^{-1} \vone} = \norm{\vec{E}} \le b$,
 hence it remains to show that $(I-H)^{-1}$ is  nonnegative.
To see this, note that $\vec{E}$ is the (unique) fixed point of a linear function $\calF$ which to every $\vec{V}$ assigns $H\cdot \vec{V}+\vec{1}$.
This function is continuous and monotone, so by Kleene's theorem we get that
 $\vec{E}=\sup_{i\in \Nset} \calF^i(\vec{0}) = \sum_{i=0}^\infty H^{i}\vone$.
Recall that $\vec{E}$ is finite, so the matrix series $H^* := \sum_{i=0}^\infty H^{i}$ converges and thus equals $(I-H)^{-1}$.
Hence $(I-H)^{-1} = H^*$, which is nonnegative as $H$ is nonnegative.

Now we are ready to apply Theorem~\ref{prop:error}.
Since $\norm{G - H} \le \varepsilon / (12\cdot b^2) < 1/v$,
 the perturbed system $\vec{V} = G \cdot \vec{V} + \vone$ has a unique solution $\vec{F}$ as desired.
By applying the second part of Theorem~\ref{prop:error} we get
 \begin{equation} \label{eq:prop-error-application}
  \frac{\norm{\vec{E} - \vec{F}}}{\norm{\vec{E}}} \ \le \ \delta \cdot 12 \cdot b \qquad \text{for $\norm{G - H} \le \delta \le 1 / (12 \cdot b)$.}
 \end{equation}
Hence,
\begin{align*}
 |E(p{\downarrow} q) - \vec{F}_{pq}|
 & \le \norm{\vec{E} - \vec{F}} && \text{(by the definition of the norm)} \\
 & \le b \cdot \frac{\norm{\vec{E} - \vec{F}}}{\norm{\vec{E}}} && \text{by $\norm{\vec{E}}\leq b$} \\%\text{(by Lemma~\ref{prop:exp-bound})} \\
 & \le b \cdot \delta \cdot 12\cdot b && \text{(by~\eqref{eq:prop-error-application})} \\
 & = \varepsilon && \text{(by the definition of~$\delta$).}
\end{align*}
%
%Now note that if the matrix $H$ is diagonal then the system $\calL'$ is in fact a system of independent
%linear equations of one variable. By solving the equations we obtain
%\[
%E(p{\downarrow}p)= \left(1-\left(P^{>0}(p,0,p)+ P^{>0}(p,1,p) \cdot
%           [p{\downarrow}p]\cdot 2\right)\right)^{-1}
%\]
%and for $p\not = q$ we get $E(p{\downarrow}q)=\left(1-P^{>0}(p,0,q)\right)^{-1}$.
%\tomas{Is this claim about precision correct?}
%Apparently by approximating
%$[p{\downarrow} p]$ up to the relative error $\xi$
%we obtain an approximation for $E(p{\downarrow}q)$ with the same relative error $\xi$.
%\tomas{Do you see any other simplification of the diagonal case? I feel that it should be
%even easier but do not see it :-)}
%
%For the rest of this proof we assume that $H$ is {\em not} diagonal.
\qed
\end{proof}

\begin{refproposition}{prop:exp-time-bound-special}
 \stmtpropexptimeboundspecial
\end{refproposition}
\begin{proof}
 The proof follows directly from Proposition~\ref{prop:grand-bound}.
\qed
\end{proof}

\subsection{Quantitative Model-Checking of  $\omega$-regular 
Properties (Section~\ref{sec-LTL})}

\begin{refproposition}{prop-product}
  Let $\Sigma$ be a finite alphabet, $\A$ a pOC, $\nu$ a valuation,
  $\R$ a DRA over $\Sigma$, and $p(0)$ a configuration of $\A$. 
  Then there is a pOC $\A'$ with Rabin acceptance condition 
  and a configuration $p'(0)$ of $\A'$ constructible in polynomial time
  such that the probability of all $w \in \run_{\A}(p(0))$ where $\nu(w)$
  is accepted by $\R$ is equal to the probability of all accepting
  $w \in \run_{\A'}(p'(0))$.   
\end{refproposition}
\begin{proof}  
Let $(E_1,F_1),\dots,(E_k,F_k)$ be the Rabin acceptance condition
of $\R$. The automaton $\A'$ is the synchronized product of 
$\A$ and $\R$  where
\begin{itemize}
\item $Q \times R$ is the set of control states, where $R$ is the set of
  states of $\R$;
\item $(p,r) \prule{x,c} (p',r')$ iff $p \prule{x,c} p'$ and
  $r \xrightarrow{\nu(p(1))} r'$ is a transition in $\R$;
\item $(p,r) \zrule{x,c} (p',r')$ iff $p \zrule{x,c} p'$ and
  $r \xrightarrow{\nu(p(0))} r'$ is a transition in $\R$.
\end{itemize}
The Rabin acceptance condition of $\A'$ is
$(Q \times E_1, Q \times F_1),\dots,(Q \times E_k,Q \times F_k)$.
\qed
\end{proof}

\begin{refproposition}{prop-visiting-approx}
  Let $c = 2|Q|$. For every $s \in G$, let $R_s$ be the probability
  of visiting a BSCC of $\G$ from $s$ in at most $c$ transitions, and
  let $R = \min\{R_s \mid s \in G\}$. Then $R > 0$ and if all transition
  probabilities in $\G$ are computed with relative error at most
  $\varepsilon R^3/8(c+1)^2$, then the resulting system
  $(I-A')\vec{V} = \vec{b}'$ has a unique solution $\vec{U}^*$ such that
  $|\vec{V}^*_s - \vec{U}^*_s|/\vec{V}^*_s \leq \varepsilon$ for every
  $s \in G$.
\end{refproposition}
\begin{proof}
The first step towards applying Theorem~\ref{thm:error} is to estimate 
the condition number $\kappa = \norm{I-A} \cdot \norm{(I-A)^{-1}}$. Obviously,
$\norm{I-A} \leq 2$. Further, $\norm{(I-A)^{-1}}$ is bounded by the
expected number of steps needed to reach a BSCC of $\G$ from a
state of $G$ (here we use a standard result about absorbing finite-state
Markov chains). 
Since $G$ has at most $c$ states, we have that $R_s > 0$, and hence
also $R >0$. Obviously, 
the probability on \emph{non-visiting} a BSCC of $\G$ in 
at most $i$ transitions from a state of $G$ is bounded
by $(1-R)^{\lfloor i/c\rfloor}$. Hence, the probability of visiting
a BSCC of $\G$ from a state of $G$ after \emph{exactly} $i$
transitions is bounded by $(1-R)^{\lfloor (i-1)/c\rfloor}$.
Further, a simple calculation shows that
\begin{eqnarray*}
  \norm{(I-A)^{-1}} & \quad\leq\quad & 
  \sum_{i=1}^\infty i \cdot (1-R)^{\lfloor (i-1)/c\rfloor}  \quad=\quad 
  \sum_{i=0}^\infty \left(\frac{c(c+1)}{2} + ic^2\right)
       \cdot \left(1-R\right)^i\\
   & = & \frac{c(c+1)}{2R} + \frac{c^2 (1-R)}{R^2}
    \quad\leq\quad \left(\frac{c+1}{R}\right)^2
\end{eqnarray*}
Hence, $\kappa \leq 2(c+1)^2/R^2$. Let $\vec{V}^*$ be the unique
solution of $(I-A)\vec{V} = \vec{b}$. Since $\norm{\vec{V}^*} \leq 1$ and 
$\vec{V}^*_s \geq R$ for every $s \in G$, it suffices to compute
an approximate solution $\vec{U}^*$ such that
\[
  \frac{\norm{\vec{V}^*-\vec{U}^*}}{\norm{\vec{V}^*}} \quad\leq\quad
  \varepsilon \cdot R
\]
By Theorem~\ref{thm:error}, we have that
\[
  \frac{\norm{\vec{V}^*-\vec{U}^*}}{\norm{\vec{V}^*}} \quad\leq\quad
  4\tau\kappa \quad\leq\quad \frac{8\tau(c+1)^2}{R^2} 
\]
where $\tau$ is the relative error of $A$ and $\vec{b}$. Hence, it
suffices to choose $\tau$ so that
\[
   \tau \quad\leq\quad 
   \frac{\varepsilon R^3}{8(c+1)^2}
\]
and compute all transition probabilities in $\G$ up to the relative
error~$\tau$.  Note that the approximation $A'$ of the matrix 
$A$ which is obtained in this way is still regular, because 
\[
  \norm{A - A'} \quad\leq\quad  \tau \quad\leq\quad 
   \frac{\varepsilon R^3}{8(c+1)^2} \quad<\quad \frac{R^2}{(c+1)^2} 
   \quad\leq\quad \frac{1}{\norm{(I-A)^{-1}}}
\]
\qed
\end{proof}

\noindent
Now we prove the divergence gap theorem. Some preliminary lemmata
are needed.

\begin{lemma} \label{lem:reach-high} Let $A$ be strongly connected and
  $t \ge 0$.  Assume $[p{\downarrow}] > 0$ for all $p \in Q$.  Let
  $\cs{0} \ge 1$ and $\ps{0} \in Q$ such that $\vv_{\ps{0}} = \vmax$.
  Let $b \in \Nset$.  Then
  \[
   \calP \left(\exists i: \cs{i} \ge b \land \forall j \le i: \cs{j} \ge 1 \;\middle\vert\; \run(\ps{0}(\cs{0})) \right)
    \quad \ge \quad \frac{1}{b + 1 + \va} \,.
  \]
\end{lemma}
\begin{proof}
 If $\cs{0} \ge b$, the lemma holds trivially.
 So we can assume that $\cs{0} < b$.
 For a run $w \in \run(\ps{0}(\cs{0}))$, we define a so-called \emph{stopping time} $\tau$ as follows:
 \[
  \tau := \inf\{ i \in \Nset_0 \mid \ms{i} \le \vmax \ \lor \ \ \ms{i} \ge b + \vmax \}
 \]
 Note that $1 + \vmax \le \ms{0} < b + \vmax$, i.e., $\tau \ge 1$.
% We will show below that $\E \tau$ is finite, which implies that $\tau$ is finite almost surely.
 Let $E$ denote the subset of runs in~$\run(\ps{0}(\cs{0}))$ where $\tau < \infty$ and $\ms{\tau} \ge b + \vmax$;
  i.e., $E$ is the event that the martingale~$\ms{i}$ reaches a value of $b+\vmax$ or higher
  without previously reaching a value of $\vmax$ or lower.
 Similarly, let $D$ denote the subset of runs in~$\run(\ps{0}(\cs{0}))$ such that
  the counter reaches a value of~$b$ ore higher without previously hitting~$0$.
 To prove the lemma we need to show $\calP(D) \ge 1 / (b + 1 + \va)$.
 We will do that by showing that $D \supseteq E$ and $\calP(E) \ge 1 / (b + 1 + \va)$.

 First we show $D \supseteq E$.
 Consider any run in~$E$;
  i.e., $\ms{\tau} \ge b + \vmax$ and $\ms{i} > \vmax$ for all $i \le \tau$.
 So, for all $i \le \tau$ we have $\ms{i} = \cs{i} + \vv_{\ps{i}} - i t > \vmax$, implying $\cs{i} > 0$.
 Similarly, $\ms{\tau} = \cs{\tau} + \vv_{\ps{\tau}} - \tau t \ge b + \vmax$, implying $\cs{\tau} \ge b$.
 Hence, the run is in~$D$, implying $D \supseteq E$.
 Hence it remains to show $\calP(E) \ge 1 / (b + 1 + \va)$.

 Next we argue that $\E \tau$ is finite:
  Since $[p{\downarrow}] > 0$ for all $p \in Q$, there are constants $k \in \Nset$ and $x \in (0,1]$ such that,
   given any configuration $p(c)$ with $p \in Q$ and $c \ge 1$, the probability of reaching in at most~$k$ steps
   a configuration $q(c-1)$ for some $q \in Q$ is at least~$x$.
  Since $A$ is strongly connected, it follows that there are constants $k' \in \Nset$ and $x' \in (0,1]$ such that,
   given any configuration $p(c)$ with $p \in Q$ and $c \ge 1$, the probability of reaching in at most~$k'$ steps
   either a configuration with zero counter or a configuration $p(c-b)$ is at least~$x'$.
  It follows that whenever $\ms{i} < b + \vmax$ the probability that there is $j \le k'$ with $\ms{i+j} \le \vmax$ is at least~$x'$.
  Hence we have
   \[
    \E \tau
    =   \sum_{\ell = 0}^\infty \calP( \tau > \ell )
    \le k' \sum_{\ell = 0}^\infty \calP( \tau > k' \ell )
    \le k' \sum_{\ell = 0}^\infty (1-x')^\ell = k' / x' \,;
   \]
  i.e., $\E \tau$ is finite.
  Consequently, the \emph{Optional Stopping Theorem}~\cite{Williams:book} is applicable and asserts
   \begin{equation}
    \E \ms{\tau} = \E \ms{0} = \ms{0} \ge 1 + \vmax\,. \label{eq:optional-stopping}
   \end{equation}

  For runs in~$E$ we have $\ms{\tau-1} < b + \vmax$.
  Since the value of $\ms{i}$ can increase by at most $1 + \va$ in a single step,
   we have $\ms{\tau} \le b + \vmax + 1 + \va$ for runs in~$E$.
  It follows that
  \begin{align*}
   \E \ms{\tau} & \le \calP(E) \cdot (b + \vmax + 1 + \va) + (1 - \calP(E)) \cdot \vmax \\
   & = \vmax + \calP(E) \cdot (b + 1 + \va) \,.
  \end{align*}
  Combining this inequality with~\eqref{eq:optional-stopping} yields $\calP(E) \ge 1 / (b + 1 + \va)$.
  This completes the proof.
\qed
\end{proof}

Let $[\ps{0}(\cs{0}){\downarrow}]$ denote the probability that a run initiated in $\ps{0}(\cs{0})$ eventually reaches counter value zero.
The following lemma gives an upper bound on~$[\ps{0}(\cs{0}){\downarrow}]$.
\begin{lemma} \label{lem:azuma}
 Let $A$ be strongly connected and $t > 0$.
 Let
  \[
   a := \exp\left(- \frac{t^2}{2 (\va + t + 1)^2} \right)\,.
  \]
 Note that $0 < a < 1$.
 Let $\cs{0} \ge \va$.
 Then we have
  \[
   [\ps{0}(\cs{0}){\downarrow}] \quad \le \quad \frac{a^{\cs{0}}}{1-a} \qquad \text{for all $\ps{0} \in Q$.}
  \]
 Moreover, if $\cs{0} \ge 6 (\va + t + 1)^3 / t^3$, then $[\ps{0}(\cs{0}){\downarrow}] \le 1/2$ for all $\ps{0} \in Q$.
\end{lemma}
\begin{proof}
 Define $H_i$ as the event that the counter reaches zero for the first time after exactly $i$ steps;
  i.e., $H_i := \{w \in \run(\ps{0}(\cs{0})) \mid \cs{i} = 0 \ \land \ \forall 0 \le j < i: \cs{j} \ge 1\}$.
 We have $[\ps{0}(\cs{0}){\downarrow}] = \calP\left( H_0 \cup H_1 \cup \cdots \right)$.
 Observe that $H_i = \emptyset$ for $i < \cs{0}$, because in each step the counter value can decrease by at most~$1$.
 For all runs in~$H_i$ we have $\ms{i} = \vv_{\ps{i}} - i t$ and so
  \begin{equation*}
   \ms{0} - \ms{i} = \cs{0} + \vv_{\ps{0}} - \vv_{\ps{i}} + i t \,.
  \end{equation*}
 It follows that
 \begin{align*}
  \calP(H_i)
  & = \calP(H_i \ \land \ \ms{0} - \ms{i} = \cs{0} + \vv_{\ps{0}} - \vv_{\ps{i}} + i t) \\
  & \le \calP(\ms{0} - \ms{i} = \cs{0} + \vv_{\ps{0}} - \vv_{\ps{i}} + i t) \\
  & \le \calP(\ms{0} - \ms{i} \ge \cs{0} - \va + i t)  \\
  & \le \calP(\ms{0} - \ms{i} \ge i t) && \text{(as $\cs{0} \ge \va$)\,.}
 \end{align*}
 In each step, the martingale value changes by at most $\va + t + 1$.
 Hence Azuma's inequality (see~\cite{Williams:book}) asserts
 \begin{align*}
  \calP(H_i)
  & \le \exp  \left(- \frac{i t^2}{2 (\va + t + 1)^2} \right) && \text{(Azuma's inequality)} \\
  & = a^i \,.
 \end{align*}
 It follows that
 \begin{align*}
  [\ps{0}(\cs{0}){\downarrow}]
    = \sum_{i=0}^\infty \calP(H_i)
  & = \sum_{i=\cs{0}}^\infty \calP(H_i) && \text{(as $H_i = \emptyset$ for $i < \cs{0}$)} \\
  & \le \sum_{i=\cs{0}}^\infty a^i && \text{(by the computation above)} \\
  & = a^{\cs{0}} / (1-a) \,.
 \end{align*}
 This proves the first statement.
 For the second statement, we need to find a condition on~$\cs{0}$ such that $[\ps{0}(\cs{0}){\downarrow}] \le 1/2$.
 The condition provided by the first statement is equivalent to
  \begin{equation*}
   \cs{0} \ge \frac{\ln(1-a) - \ln 2}{\ln a}\,. \label{eq:cond-threshold}
  \end{equation*}
 Define $d := \frac{t^2}{2 (\va + t + 1)^2}$.
 Note that $a = \exp(-d)$ and $0 < d < 1$.
 It is straightforward to verify that
  \[
   \frac{\ln(1-\exp(-d)) - \ln 2}{-d} \le \frac{2}{d^{3/2}} \quad \text{for all $0 < d < 1$.}
  \]
 Since
  \[
   \frac{2}{d^{3/2}} = \frac{2 \cdot 2^{3/2} \cdot (\va + t + 1)^3}{t^3} \le \frac{6  (\va + t + 1)^3}{t^3}\,,
  \]
 the second statement follows.
\qed
\end{proof}

\begin{proposition} \label{prop:gap}
 Let $A$ be strongly connected and $t > 0$ and $[p{\downarrow}] > 0$ for all $p \in Q$.
% If no configuration $p(c)$ with $\vv_p = \vmax$ and $c \ge 1$ can be reached from $\ps{0}(1)$,
%  then $[\ps{0}{\uparrow}] = 0]$.
% Otherwise,
 Let $p \in Q$ with $\vv_p = \vmax$.
 Then
 \[
  [p{\uparrow}] \ge \frac{t^3}{12 (2 \va + 4)^3}\,.
 \]
\end{proposition}
\begin{proof}
% Since $A$ is strongly connected, a run in $\run(\ps{0}{\uparrow})$ visits all states infinitely often almost surely.
% In particular, it visits configurations $p(c)$ with $\vv_p = \vmax$ and $c \ge 1$ infinitely often almost surely.
% So, if such configurations are not reachable, we have $[\ps{0}{\uparrow}] = 0]$.
%
%For any $b \in \Nset$, consider the event~$E_b \subseteq \run(p(1))$ of reaching a configuration~$q(b)$ for some $q \in Q$
% without the counter hitting~$0$, and thereafter never reaching a counter value below~$b$.
%Note that $E_b \subseteq \run(p{\uparrow})$.
%So, to prove the proposition, it suffices to show $\calP(E_b) \ge \frac{t^3}{12 (2 \va + 4)^3}$.
Define $b$ as the smallest integer $b \ge 6 (\va + t + 1)^3 / t^3$.
By Lemma~\ref{lem:reach-high} we have
  \[
   \calP \left(\exists i: \cs{i} \ge b \land \forall j \le i: \cs{j} \ge 1 \;\middle\vert\; \run(p(1)) \right)
    \quad
     \ge \quad \frac{1}{b + 1 + \va} \,.
  \]
Since $0 < t \le 1$, we have
 \[
  b + 1 + \va \quad\le\quad 6 (\va + t + 2)^3 / t^3 + 1 + \va \quad\le\quad 6 (2 \va + 4)^3 / t^3
 \]
and so
 \[
   \calP \left(\exists i: \cs{i} \ge b \land \forall j \le i: \cs{j} \ge 1 \;\middle\vert\; \run(p(1)) \right)
    \quad \ge \quad \frac{t^3}{6 (2 \va + 4)^3}\,.
 \]
Using the Markov property and Lemma~\ref{lem:azuma} we obtain
 \[
  [p{\uparrow}] \ge \frac{t^3}{12 (2 \va + 4)^3} \,.
 \]
\qed
\end{proof}
Now let us drop the assumption that $A$ is strongly connected. Each 
BSCC $\B$ of $A$ induces
a strongly connected pOC in which we have a trend $t$ and 
a potential $\vec{v}$.
\begin{reftheorem}{thm-gap}
  Let $\A = (Q,\delta^{=0},\delta^{>0},P^{=0},P^{>0})$ be a pOC and
  $\X$ the underlying finite-state Markov chain of $\A$. 
  Let $p \in Q$ such that $[p{\uparrow}]>0$. Then there are two 
  possibilities:
  \begin{enumerate}
     \item There is $q\in Q$ such that $[p,q]>0$ and $[q{\uparrow}]=1$.
        Hence, $[p{\uparrow}] \geq [p,q]$.
     \item There is a BSCC $\B$ of $\X$ and a state $q$ of $\B$ such 
        that $[p,q]>0$, $t > 0$, and $\vec{v}_{q}=\vec{v}_{\max}$
        (here $t$ is the trend, $\vec{v}$ is the vector
        of Proposition~\ref{prop-martingale}, and $\vec{v}_{\max}$ 
        is the maximal component of~$\vec{v}$; all of these
        are considered in $\B$). Further,
        \[
            [p{\uparrow}]\quad \ge\quad 
            [p,q]\cdot \frac{t^3}{12 (2 |\vec{v}| + 4)^3}\,.
        \]       
 \end{enumerate}
\end{reftheorem}
% \begin{corollary} \label{prop:gap}
%  %Denote $t:=\min\{ t^C\mid t^C>0\}$ and $v:=\max_{C} |\vec{v}^C|$.
%  Assume that $[p{\uparrow}]>0$. Then one of the following is true:
%  \begin{enumerate}
%  \item There is $q\in Q$ such that $[p,q]>0$ and $[q{\uparrow}]=1$.
%  \item There is a BSCC $C$ and a state $q\in C$ such that $[p,q]>0$ and 
%  $t^{C}>0$ and $\vec{v}^{C}_{q}=\vec{v}^{C}_{\max}$.
%  \end{enumerate}
%  If 1. is true, then $[p{\uparrow}]=[p,q]$. If 2. is true, then
%  \[
%   [p{\uparrow}]\quad \ge\quad [p,q]\cdot \frac{\left(t^C\right)^3}{12 (2 |\vec{v}^C| + 4)^3}\,.
%  \]
% % 
% % 
% % If there is $q\in Q$ such that $[p,q]>0$ and $[q{\uparrow}]=1$, then $[p{\uparrow}]=[p,q]$. Otherwise, if there is no such $q\in Q$, then
% % there is a BSCC $C$ and a state $q\in C$ such that $[p,q]>0$ and 
% % $t^{C}>0$ and $\vec{v}^{C}_{q}=\vec{v}^{C}_{\max}$;
% % in which case 
% % \[
% %  [p{\uparrow}]\quad \ge\quad [p,q]\cdot \frac{\left(t^C\right)^3}{12 (2 |\vec{v}^C| + 4)^3}\,.
% % \]
% % the following is true.
% % %Let $A$ be strongly connected and assume that $[p{\downarrow}]>0$ for all $p\in Q$. 
% % \begin{itemize}
% % \item Either If $t > 0$, then for every $p\in Q$ satisfying $[p{\uparrow}]>0$ we have
% % \[
% %  [p{\uparrow}]\quad \ge\quad \breach\cdot \frac{t^3}{12 (2 \va + 4)^3}\,.
% % \]
% % \item If $t \leq 0$, then $[p{\uparrow}]=0$.
% % \end{itemize}
% %\tomas{Here $\breach$ is the minimum of all probabilities $[p,q]$.}
% \end{corollary}
\begin{proof}
  Assume that $[q{\uparrow}]<1$ for all $q\in Q$. Given a BSCC $\B$,
  denote by $R_{\B}$ the set of runs of $\run(p{\uparrow})$ that reach
  $\B$. Almost all runs of $\run(p{\uparrow})$ belong to $\bigcup_{\B}
  R_{\B}$. Moreover, using strong law of large numbers
  (see~e.g.~\cite{Williams:book}) and results of~\cite{BBEKW:OC-MDP-arXiv} (in
  particular Lemma~19), one can show that almost every run of
  $\run(p{\uparrow})$ belongs to some $R_{\B}$ satisfying $t>0$.  It
  follows that there is a BSCC $\B$ such that $t>0$ and
  $\calP(R_\B)>0$.  Now almost all runs of $R_\B$ either terminate, or
  visit all states of $\B$ infinitely many times. In particular, almost
  all runs of $R_\B$ reach a state $q$ satisfying
  $\vec{v}_{q}=\vec{v}_{\max}$, and thus $[p,q]>0$.  \qed
\end{proof}

%%% Local Variables: 
%%% mode: latex
%%% TeX-master: "main"
%%% End: 

%\input{app-notation}
%\input{app-martingale}
%\input{app-condition}
%\input{app-gap}

\end{document}